%% file: main.tex
\documentclass[runningheads]{llncs}

\usepackage{ amssymb }
\usepackage{stmaryrd}

\usepackage[T1]{fontenc}
\usepackage[utf8]{inputenc}
\usepackage[british]{babel}
\usepackage{xspace, listings, lstcustom, wrapfig, graphicx, enumerate}
\usepackage{paralist}
\usepackage{color,colortbl, relsize}
\usepackage{rotating}
\usepackage{pifont}
\usepackage{multirow}
\usepackage{soul}
\usepackage{tcolorbox}
\usepackage[scaled=.9, light]{zlmtt}
\usepackage{siunitx}
\usepackage{setspace}
\usepackage{enumitem}

\include{macros}

\usepackage{times}
\usepackage{latexsym}
\usepackage{listings}
\definecolor{dkgreen}{rgb}{0,0.6,0}
\definecolor{gray}{rgb}{0.5,0.5,0.5}
\definecolor{mauve}{rgb}{0.58,0,0.82}

\newcommand{\secomment}[1]{{\ensuremath{\blacksquare}}\footnote{\se{#1}}}

 \newcommand{\sd}[1]{#1} 
\newcommand{\se}[1]{#1} 

\newcommand{\sophia}[1]{#1} 
\newcommand{\susan}[1]{#1} 
\newcommand{\toby}[1]{} 
\newcommand{\james}[1]{} 

\begin{document}

\title{Holistic Specifications for Robust Programs}


\author{Sophia Drossopoulou\inst{1}\orcidID{0000-0002-1993-1142} \and
  James Noble\inst{2}\orcidID{0000-0001-9036-5692} \and
 Julian Mackay\inst{2}\orcidID{0000-0003-3098-3901} \and
Susan Eisenbach\inst{1}\orcidID{0000-0001-9072-6689}}

\authorrunning{S. Drossopoulou, J. Noble, et al.}

\institute{Imperial College London
\email{\{scd,susan\}@imperial.ac.uk} \and
Victoria University of Wellington
\email{\{julian.mackay,kjx\}@ecs.vuw.ac.nz}}

\maketitle

\begin{abstract}
Functional specifications describe what
program components \emph{can} do: the \emph{sufficient} conditions to
invoke a component's operations.
They allow us to reason about the  use  of  
components in  the \emph{closed world} setting, where
the component interacts with known client code, and 
where the client code must establish the appropriate pre-conditions 
before calling into the component.
%

Sufficient conditions are not enough to reason about 
the use of components in the \emph{open world} setting, where
the component interacts with external code, possibly of unknown
provenance, 
and where the  component itself may evolve over time. 
In this   open world  setting,
we must also consider the
\emph{necessary} conditions, \ie
what are the conditions without which an effect will \emph{not} happen. 
In this paper we
 propose a 
 language \Chainmail for writing \emph{holistic} specifications that
 focus on necessary conditions (as well as sufficient conditions). We
 give a formal semantics for \Chainmail, and the core of \Chainmail has been
  mechanised in the Coq proof assistant.
\end{abstract}


\section{Introduction}
\input{introductionWithWallet}

\section{Motivating Example: The Bank}
\label{sect:motivate:Bank}

\input{motivateBankShort}
\section{\Chainmail\ Overview}
\label{sect:chainmail}
\input{overChainmail}

\section{Overview of the Formal foundations}
\label{sect:formal}
We now give an overview of the formal model for \Chainmail. In section \ref{sect:overviewmodel} we 
introduce  the shape of the judgments used to give semantics to \Chainmail, while in section \ref{sect:PL} 
we describe the most salient aspects of an underlying programming language used in \Chainmail.

\subsection{\Chainmail\ judgments}
\label{sect:overviewmodel}
\input{overModel}

\subsection{An underlying programming language, \LangOO}
\label{sect:PL}
\input{summaryExecution}

\section{Assertions}
\label{sect:assertions}

\input{assertions}








\section{Examplar Driven Design}
\label{sect:discussion} 
\input{example-driven-design}

\section{Related Work}
\label{sect:related}
\input{related}
\section{Conclusions}
\label{sect:conclusion}
\input{conclusion}
\section{Acknowledgments}

\sophia{This work is based on a long-standing collaboration with Mark M. Miller and Toby Murray.
We have received invaluable feedback from Alex Summers, Bart Jacobs, Michael Jackson, members of WG 2.3, 
and the FASE 2020  reviewers.
The work has been supported by the NZ Marsden grants xxxx, and yyy, and also research gifts from Agoric\susan{, the Ethereum Foundation,} and
Facebook.}

\bibliographystyle{splncs04}
\bibliography{Case,more}

\newpage
\appendix

\section{Foundations: Programming Language}
\label{app:LangOO}
\input{langOO}

\clearpage

\section{Foundations: Specification Language}
\label{app:assertions}
\input{assertionsAppendix}

\clearpage

\section{Exemplar: Bank Account}
\label{Bank:appendix}
\input{Bank.appendix}

\clearpage

\section{Examplar: Authorising ERC20}
\label{sect:example:ERC20}
{\input{ERC20}}

\clearpage

\subsection{Example -- ERC20, the traditional specification }
\label{ERC20:appendix}
\input{ERC20.appendix.tex}

\clearpage

\section{Examplar: Defending the DAO}
\label{Dao:appendix}
\input{DAO}

\clearpage

\section{Examplar: Attenuating the DOM}
\label{sect:example:DOM}
\input{DOM}

\clearpage

\section{Coq Formalism}
\label{sect:coq}
\input{COQ}


\end{document}

%% file: macros.tex
\newcommand{\sdparagraph}[1]{{\vspace{.3cm} \noindent \textit{#1}\ }}

\newcommand{\prg}[1]{{\mbox{\tt{#1}}}}

\newcommand{\m}{\prg{m}}
 \newcommand{\f}{\prg{f}}
  \newcommand{\e}{\prg{e}}
 \renewcommand{\c}{\prg{C}}
 \renewcommand{\v}{\prg{v}}
  \newcommand{\x}{\prg{x}}
  \newcommand{\p}{\prg{p}}
   \newcommand{\y}{\prg{y}}
      \newcommand{\uu}{\prg{u}}
  \newcommand{\this}{\prg{this}}

\newcommand{\fldMap}{\textit{fldMap}}

\newcommand{\forget}[1]{}
\newcommand{\etc}{{\it etc.}}
\newcommand{\eg}{{\it e.g.\,}}
\newcommand{\ie}{{\it i.e.\,}}

\newcommand{\Future}[1] {\ensuremath{{\mathsf{will}}\langle \,#1\,\rangle}}
\newcommand{\Using}[2]{\ensuremath{\langle\, #1\, \mathsf{in}\, #2\, \rangle}}
\newcommand{\External}[1] {\ensuremath{{\mathsf{external}}\langle\,  #1\, \rangle}}
\newcommand{\Internal}[1] {\ensuremath{{\mathsf{internal}}\langle\,  #1\, \rangle}}
\newcommand{\Changes}[1]{\ensuremath{{\mathsf{changes}}\langle\,#1\,\rangle}}
\newcommand{\CanAccessTr}[2]{\ensuremath{\langle\, #1\, \mathsf{access}^*\, #2\, \rangle}}
\newcommand{\CanAccess}[2]{\ensuremath{\langle\, #1\, \mathsf{access}\, #2\, \rangle}}

\newcommand{\Calls}[4]{\ensuremath{\langle\, #1\, \mathsf{calls}\, #3.#2\lp#4\rp\, \rangle}}
\newcommand{\Next}[1] {\ensuremath{{\mathsf{next}}\langle \,#1\,\rangle}}
\newcommand{\PrevId}{\ensuremath{{\mathcal{P}}\textrm{\textit{rev}}}}
\newcommand{\Prev}[1] {\ensuremath{{\mathsf{prev}}\langle \,#1\,\rangle}}
\newcommand{\Past}[1]  {\ensuremath{{\mathsf{was}}\langle \,#1\,\rangle}}

\newcommand{\Initial}[1] {{{\mathcal I}\!nitial}\langle#1\rangle}

\newcommand{\appref}[1]{see App.~\ref{#1}}

\newcommand{\varMap}{{\ensuremath{\beta}}}

\newcommand{\LangOO} {\ensuremath{{\mathcal L}{_{\tt {oo}}}}\xspace}
\newcommand{\Chainmail} {\ensuremath{{\mathcal C}hainmail}\xspace}

\newcommand{\cf}{{\it c.f.~}}

\newcommand{\SP}{{\hspace{.1in}}}

\newcommand{\A}{\ensuremath{A}}

\newcommand{\Arising}[1]{{\mathcal{A}}\textrm{\textit{rising}}(#1)}

\newcommand{\syntax}[1]{\prg{{\it #1}}}
\newcommand{\BBC}{$::=$} 
\newcommand{\SOR}{\ensuremath{\ \mid\ }} 

\newcommand{\pre}{\ensuremath{_{{pre}}}}   

 \newcommand{\interp}[2]{{\ensuremath{\lfloor{ {#1}}\rfloor_{#2}}}}

\newcommand{\kw}[1]{\prg{#1}} 
\newcommand{\kwN}[1]{{\bf{\sf {#1}}}}
\newcommand{\returnKW}{{\bf{\sf {return}}}}
\newcommand{\newKW}{\mbox{\bf{\sf{new}}}}

\newcommand{\semi}{\mbox{{\kw {;}}\ }}

\newcommand{\lb}{\prg{\mbox{\tt{\bf{\{ }}}}}
\newcommand{\rb}{\prg{\mbox{\tt{\bf{\} }}}}}
\newcommand{\lp}{\prg{\mbox{\tt{\bf{( }}}}}
\newcommand{\rp}{\prg{\mbox{\tt{\bf{) }}}}}

 \newcommand{\M}{\prg{\ensuremath{\prg{M}}}}
  \newcommand{\Prog}[1]{\M{#1}}

\newcommand{\mkpair}{\fatsemi}

\newcommand{\restrct}[2]{\ensuremath{#1\!\!\downarrow\!_{#2}}}
\newcommand{\adapt}{\ensuremath{\!\triangleleft\!}}
\newcommand{\link}{\!\circ\!}

\newcommand{\ClassOf}[2] {\ensuremath{{\mathcal C}{\mathit{lass}}(#1)_{#2}}}

\newcommand{\SE}{\ensuremath{\prg{e}}}
\newcommand{\SEPrime}{\ensuremath{\prg{e}'}}

\newcommand{\expandexp}[1]{}

 \newcommand{\IFF}{{\SP {\mbox{ if }} \SP}}

\newcommand{\SAF}{\ensuremath{W}} 

\newcommand{\inferenceruleN}[3]
{
\begin{array}{l}
\SP\SP\SP\SP\SP  {\sf #1}
\\ #2  \\ \hline   #3
  \end{array}
}

\newcommand{\inferenceruleNM}[3]
{
\begin{array}{l}
\SP\SP\SP\SP\SP \SP\SP\SP\SP\SP\SP\SP\SP  {\sf #1}
\\ #2  \\ \hline   #3
  \end{array}
}

\newcommand{\inferenceruleNN}[3]
{
\begin{array}{l}
\SP\SP\SP\SP\SP\SP\SP\SP
\SP\SP\SP\SP\SP\SP\SP\SP
\SP\SP\SP\SP\SP\SP\SP\SP
\SP\SP\SP\SP\SP\SP\SP\SP
\SP\SP\SP\SP\SP\SP\SP\SP
\SP\SP
   {\sf #1}
\\ #2  \\ \hline   #3
  \end{array}
}

\newif\ifNumberResults\NumberResultstrue
\def\@@opargbegintheorem#1#2#3{\@@@@begintheorem{\bf\@@thmname{#1}{#2}(#3)}}
\def\@@begintheorem#1#2{\@@@@begintheorem{\bf\@@thmname{#1}{#2}}}
\def\@@@@begintheorem#1{\par\removelastskip\smallskip\noindent{#1}}
\def\@@thmname#1#2{#1\ \ifNumberResults#2\ \fi}

\newcommand{\SF}{{\prg S}}
\newcommand{\z}{{\prg z}}
\newcommand{\pu}{{\prg u}}
\newcommand{\pb}{{\prg b}}
\newcommand{\acc}{{\prg a}}
\newcommand{\bal}{{\prg{balance}}}
\newcommand{\zs}{{\prg{zs}}}

\newcommand{\Meths}[3]{\ensuremath{{\mathcal M}(}\Prog{#1},\prg{#2},
\prg{#3}\ensuremath{)} }

\newcommand{\WideFigWhere}[4] 
{
\begin{figure*}[{#4}]
\begin{center}
\noindent
\fbox{
\begin{minipage}{5. in}
{#1} 
\end{minipage}
}
\caption{#2}
\label{#3}
\end{center}
\end{figure*}
}

\newcommand{\BigWideFigWhere}[4] 
{
\begin{figure*}[{#4}]
\begin{center}
\noindent
{\normalsize
\hrule
\begin{minipage}{5. in}
{#1} 
\end{minipage}
\hrule
}
\caption{#2}
\label{#3}
\end{center}
\end{figure*}
}

\newcommand{\NotTooWideFigWhere}[4] 
{
\begin{figure*}[{#4}]
\begin{center}
\noindent
\fbox{
\begin{minipage}{4.3 in}
{#1} 
\end{minipage}
}
\caption{#2}
\label{#3}
\end{center}
\end{figure*}
}

\newcommand{\BigNotTooWideFigWhere}[4] 
{
\begin{figure*}[{#4}]
\begin{center}
\noindent
{\normalsize
\hrule
\begin{minipage}{4.3 in}
{#1} 
\end{minipage}
\hrule
}
\caption{#2}
\label{#3}
\end{center}
\end{figure*}
}


%% file: introductionWithWallet.tex
Software guards our secrets, our money, our intellectual property,
our reputation \cite{covern}.  We entrust personal and
corporate information to software which works in an \emph{open} world, 
where  it interacts with 
third party software of unknown provenance, possibly buggy and potentially malicious.

This means we need our software to be \emph{robust}:
to behave correctly even if  used 
by erroneous or malicious third parties.
We expect that our bank will only make payments 
from our account if instructed by us, or by somebody we have authorised, 
that space on a web given to an advertiser will not be used
to obtain access to our bank details \cite{cwe}, or that a
concert hall will not book the same seat more than once.


While language mechanisms such as constants, invariants, 
object capabilities \cite{MillerPhD}, and 
ownership \cite{ownalias} 
make it \textit{possible} to write robust
programs, they cannot \textit{ensure} that programs are robust.
Ensuring robustness is difficult because it means 
different things for different systems: perhaps
that critical operations should only be invoked with the requisite authority;
perhaps that sensitive personal information should not be leaked; 
or perhaps that a resource belonging to one user should not be consumed by another.
To ensure robustness, we need ways to specify what robustness means for a 
particular program, and ways to demonstrate that the particular program 
adheres to its specific robustness requirements.

 \begin{figure}[htb]
 \begin{tabular}{lll} 
\begin{minipage}{0.45\textwidth}
\begin{lstlisting}
class Safe{
   field treasure 
   field secret 
   method take(scr){
      if (secret==scr) then {
         t=treasure
         treasure = null
         return t }  }
 }
\end{lstlisting}
\end{minipage}
  &\ \ \  \ \ \ \ \  \ \ \ \ \ \ &
\begin{minipage}{0.45\textwidth}
\begin{lstlisting}
class Safe{
   field treasure   
   field secret  
   method take(scr){
       $\mathit{... as\, version\,1 ...}$ 
   }
   method set(scr){
         secret=scr }
 }
\end{lstlisting}
\end{minipage} 
 \end{tabular}
  \vspace*{-0.95cm}
  \caption{Two Versions of the class \prg{Safe}}
 \label{fig:ExampleSafe}
 \vspace*{-0.65cm}
 \end{figure}
 

Consider the code snippets from Fig. \ref{fig:ExampleSafe}. Objects of
 class \prg{Safe}  hold a \prg{treasure} and a \prg{secret}, and  only
 the holder of the secret can remove the treasure from the safe.
 We show the code in two versions; both   have the same method \prg{take}, and the second version 
 has an additional method \prg{set}.
  We 
 assume \sophia{a dynamically} typed language (\sophia{so that our results are applicable to the statically
 as well as the dynamically typed setting}); 
 \secomment{We do not depend on the additional safety static typing provides, so we only assume a dynamically typed language.}
\sophia{that} fields are private in the sense of Java \sophia{(}\ie only
 methods of that class may read or write these fields);
 and that addresses are unforgeable \sophia{(}so  there is no way to guess a secret). 
 A classical Hoare triple describing the behaviour of \prg{take} would be:
 
  \vspace{.1in}
  
(ClassicSpec)$  \ \ $  $\triangleq$
\vspace{-.1in}
\begin{lstlisting}
   method take(scr)
   PRE:   this:Safe  
   POST:  scr=this.secret$\pre$  $\longrightarrow$ this.treasure=null 
               $\wedge$
          scr$\neq$this.secret$\pre$ $\longrightarrow$  $\forall$s:Safe.$\,$s.treasure=s.treasure$\pre$
 \end{lstlisting}
\vspace{-.2in}

(ClassicSpec)  expresses  that knowledge of the \prg{secret} is  \emph{sufficient} 
to remove the treasure, and that  
%
 \prg{take} cannot remove the treasure unless the secret is
provided. 
But it cannot preclude that \prg{Safe} -- or some other class, for that matter -- contains more methods 
which might make it possible to remove the treasure  without knowledge of the
secret.
\sophia{This is the problem with the second version of \prg{Safe}: it satisfies (ClassicSpec)}, but is 
not robust, \sd{as it is possible to overwrite the \prg{secret} of the \prg{Safe} and then use it to
remove the treasure.}
To express robustness requirements, we introduce \emph{holistic specifications}, and require that:
 
  \vspace{.01in}
(HolisticSpec)\ \  $\triangleq$\\ 
$\strut \hspace{.3in}   \forall \prg{s}. 
[\ \ \prg{s}:\prg{Safe} \wedge \prg{s.treasure}\neq\prg{null}   \wedge   \Future{\prg{s.treasure}=\prg{null}} $ \\ 
 $ \strut \hspace{5.3cm}     \longrightarrow \ \  \exists \prg{o}. [\  \External{\prg{o}} \wedge  \CanAccess{\prg{o}}{\prg{s.secret}}\ ]  \  \ ] \hfill $
\vspace{-.05in}

(HolisticSpec) mandates that for any safe \prg{s}
whose treasure is not \prg{null}, 
if some time in the future its treasure were to become \prg{null},
then at least one external object (\ie an object whose class is not \prg{Safe}) in the current configuration
has direct access to \prg{s}'s secret. This external object need not have caused the change in $\prg{s.treasure}$ but it would 
 have (transitively) passed access to the secret which ultimately did cause that change.
Both classes in Fig. \ref{fig:ExampleSafe} satisfy (ClassicSpec), but the second version does not satisfy 
(HolisticSpec).
%
%
\vspace{.02in}
%
%

In this paper we propose \Chainmail, a specification language to
express holistic specifications.
The design of \Chainmail was guided by the study of a sequence of
examples from the object-capability literature and the smart contracts world: the
membrane \cite{membranesJavascript}, the DOM \cite{dd,ddd}, the Mint/Purse \cite{MillerPhD}, the Escrow \cite{proxiesECOOP2013}, the DAO \cite{Dao,DaoBug} and
ERC20 \cite{ERC20}.  As we worked through the
examples, we found a small set of language constructs that let us
write holistic specifications across a range of different contexts.
\Chainmail extends 
traditional program specification languages \cite{Leavens-etal07,Meyer92} with features which talk about:
\begin{description}
\item[Permission: ] 
Which objects may have access to which other objects; 
this is central since access to an object grants access to the functions it provides.
\item[Control: ]
Which objects called functions on other objects; this
 is useful in identifying the causes of certain effects - eg 
funds can only be reduced if the owner called a payment function.
%
%
\item[Time: ]
What holds some time in  the past, the future, and what changes with time,
\item[Space: ]
Which parts of the heap are considered when establishing some property, or when 
performing program execution; a concept
related to, but different from, memory footprints and separation logics,
\item[Viewpoint: ]
Which objects and which configurations are internal to our component, and which  are
external to it;
a concept related to the open world setting.
\end{description}


While many individual features of \Chainmail can be found in other work, 
their power and novelty for specifying open systems lies in their careful combination.
The contributions of this paper are:
\begin{itemize}
\item the design of the holistic specification language \Chainmail,
\item the semantics of \Chainmail, 
\item a Coq mechanisation of the core of \Chainmail.
\end{itemize}

The rest of the paper is organised as follows:
Section~\ref{sect:motivate:Bank} 
\sd{gives an example from the literature} which we will use 
to elucidate key points of \Chainmail.
~\ref{sect:chainmail} presents the \Chainmail\ specification
language.  Section~\ref{sect:formal} introduces the formal model
underlying \Chainmail, and then section~\ref{sect:assertions} defines
the 
semantics of \Chainmail's assertions.
Section~\ref{sect:discussion}
discusses our design, \ref{sect:related} considers related
work, and section~\ref{sect:conclusion} concludes.
We relegate key points of exemplar problems and various details to appendices which are available at~\cite{examples}.

%% file: motivateBankShort.tex
As a motivating example, we consider a simplified banking application 
taken from the object capabilities literature \cite{ELang}:
 \prg{Account}s belong to \prg{Bank}s and hold money (\prg{balance}s);  
with access  to two \prg{Account}s of the same  \prg{Bank} one can  transfer any amount of money from
 one to the other.  
This example has the advantage that it requires several objects and classes.

%
%

We will not show the code  here (see appendix~\ref{Bank:appendix}), but suffice it to say that class 
\prg{Account} has  methods \prg{deposit(src, amt)} and 
\prg{makeAccount(amt)}  (\ie a method called \prg{deposit}
with two arguments, and a method called \prg{makeAccount}
with one argument). Similarly, \prg{Bank} has method  \prg{newAccount(amt)}.
Moreover, \prg{deposit} requires that the receiver and
first argument  (\prg{this} and \prg{src}) are \prg{Account}s
and belong to the same bank,
that the second argument (\prg{amt}) is a number, and that \prg{src}'s
balance is at least \prg{amt}.
If this condition  holds, then 
 \prg{amt} gets transferred from \prg{src} to the receiver.
 The function \prg{makeNewAccount}  returns a fresh \prg{Account} with the same bank, and transfers \prg{amt}
 from the receiver \prg{Account} to the new \prg{Account}.
 Finally, the function \prg{newAccount} when run by a \prg{Bank} creates a new \prg{Account} with corresponding 
 amount of money in it.\footnote{{Note that our very limited bank specification doesn't even have the concept of an account owner.}}
It is not difficult to give formal specifications of these methods in terms of 
pre- and post-conditions. 

However,   what if the bank provided a \prg{steal} method that 
 emptied out every account in the bank into a thief's account?
%
%
%
The critical problem is that a bank implementation
 including a \prg{steal}
method could meet the functional specifications of 
\prg{deposit}, \prg{makeAccount}, and \prg{newAccount},
and still allow the clients' money to be stolen.
%

One obvious solution would be to adopt 
a closed-world
interpretation of specifications: we interpret \sd{functional} specifications  
as \emph{exact} in the sense that only
implementations that meet the functional specification exactly,
\emph{with no extra methods or behaviour}, are considered as suitable
implementations of the functional specification. The problem is that
this solution is far too strong: it would for example rule out a bank
that  during software maintenance was given a new method \prg{count}
that simply counted the number of deposits that had taken place, or a method \prg{notify}
to enable the bank to occasionally send notifications  to its customers.
%
%

What we need is some way to permit bank implementations that 
send notifications to customers, but to forbid implementations of \prg{steal}. 
The key here is to capture the (implicit)
assumptions underlying 
the design of the banking application.
 We provide
additional specifications that capture those assumptions.  The following
 three informal requirements   prevent methods like \prg{steal}:

\begin{enumerate}
\item 
An account's
  balance can be changed only  if a client   calls the \prg{deposit} method  with the
  account as the receiver or as an argument. 
\item An account's balance can be changed  only  if a client has
  access to that  particular account. 
\item The \prg{Bank}/\prg{Account} component does not leak access to existing accounts or banks. 
\end{enumerate}

Compared with the functional specification we have seen so far, these
requirements 
capture \emph{necessary} 
rather than
\emph{sufficient} conditions:  Calling the \prg{deposit}
method to gain access to an account 
is necessary for any change to that account taking place.
The  function 
\prg{steal} is inconsistent with requirement  (1), as it reduces the balance of an \prg{Account} without calling the
function \prg{deposit}. 
However, requirement  (1) is not enough to protect our money. We need to (2) to avoid an \prg{Account}'s balance getting
modified without access to the particular \prg{Account}, and (3) to ensure that such accesses are not leaked. 


We can  express these  requirements  
through \Chainmail assertions.  Rather than 
specifying the behaviour of particular methods when they are called, we
write  assertions   that range across the entire behaviour of the
\prg{Bank}/\prg{Account}  module.
\vspace{.2cm}


(1)\ \  $\triangleq$\ \ $\forall \prg{a}.[\ \ \prg{a}:\prg{Account} \wedge \Changes{\prg{a.balance}}  \ \    
    \longrightarrow \ \    \hfill$ \\
  $\strut \hspace{2.3cm} 
  \exists \prg{o}. [\    \Calls{\prg{o}}{\prg{deposit}}{\prg{a}}{\_,\_} \vee\  \Calls{\prg{o}}{\prg{deposit}}{\_}{\prg{a},\_}\  \ ] \ \ \ \ ] \hfill $

\vspace{.4cm}

    (2)\ \  $\triangleq$\ \ $\forall \prg{a}.\forall \prg{S}:\prg{Set}.\ [  \ \  \prg{a}:\prg{Account}\ \wedge \  \Using{ \Future{\,\Changes{\prg{a.balance}}} \,}{\prg{S}\,}\ \ \   \
    \longrightarrow$ \\
 $\strut \hspace{3.9cm} \hfill \exists \prg{o}.\ [\, \prg{o}\in \prg{S}\ \wedge\  \External{\prg{o}}  \ \wedge \ \CanAccess{\prg{o}}{\prg{a}} \, ] \ \ \ \ ]$
\vspace{.4cm} 
 
     (3)\ \  $\triangleq$\ \ $\forall \prg{a}.\forall \prg{S}:\prg{Set}.\ [ \ \  \prg{a}:\prg{Account}\ \wedge \ \ {\Using{\Future{\ \exists \prg{o}.[\ \External{\prg{o}} \ \wedge\ \CanAccess{\prg{o}}{\prg{a}}]}}{\SF}}$ \\  
 $\strut \hspace{3.9cm} \hfill   \longrightarrow \ \ \ \exists \prg{o}'.\ [\, \prg{o}'\in \prg{S}\  \wedge  \ \External{\prg{o}'}  \ \wedge \ \CanAccess{\prg{o}'}{\prg{a}}   \ ] \ \ \ \ ]$

\vspace{.2cm}

\noindent 
In the above and throughout the paper, we use an underscore ($\_$) to indicate an existentially bound variable whose 
value is of no interest.

\vspace{.2cm}

Assertion (1) 
says that if   an account's balance changes
($\Changes{\prg{a.balance}}$),
then there must be some client object \prg{o}
that 
called the \prg{deposit} method with \prg{a} as a receiver or as an argument 
($\Calls{\prg{o}} {\prg{deposit}} {\_} {\_}$).
 
Assertion (2) similarly constrains any possible change to an 
account's balance: 
If at some future point the balance changes  (${\Future{\,\Changes{...}}}$),  
and if this future change is observed with the state restricted to the objects from \SF~ (\ie $\Using{...}{\prg{S}}$), then 
at least one of these objects ($\prg{o}\in\SF$) is external to the \prg{Bank}/\prg{Account} system ($\External{\prg{o}}$) and 
has (direct) access to that account object
($\CanAccess{\prg{o}}{\prg{a}}$).
Notice that while the change in the \prg{balance} happens some time in the future,
the external object \prg{o} has access to \prg{a} in the \emph{current} state.
Notice also, that the object which makes the call to \prg{deposit} described in (1), and the object which 
has access to \prg{a} in the current state described in (2) need not be the same: It may well be that the
latter passes  a reference to \prg{a} to the former (indirectly), which then makes the call
to \prg{deposit}.

It remains to think about how access to an \prg{Account} may be obtained. This is the remit of assertion (3): 
It says that if at some time in the future of the state restricted to \SF, 
some object \prg{o} which is external has access to some account \prg{a}, and if \prg{a} exists in the 
current state, then in the current state some object 
from \SF~has access to \prg{a}. Where \prg{o} and $\prg{o}'$ may, but need not, be the same object. And where
 $\prg{o}'$ has to exist and have access to \prg{a} in the \emph{current} state, but 
\prg{o} need not exist in the current state -- it may be allocated later.
Assertion (3) thus gives essential protection when dealing with foreign, untrusted code.
When an \prg{Account} is given out to untrusted third parties, assertion (3) guarantees that
this \prg{Account} cannot be used to obtain access to further  \prg{Account}s. 

\vspace{.1cm}

A  holistic  specification for the bank account, then,
would be a sufficient functional specification
plus the necessary
specifications (1)-(3) from above. 
This holistic specification
permits an implementation of the bank that also provides  \prg{count}
and \prg{notify} methods, even though the specification does not mention either method.
Critically, though, the \Chainmail specification
does not permit an
implementation that includes a \prg{steal} method.

%% file: overChainmail.tex
%
In this Section we  give a brief and informal  overview of some of the most salient features of  
\Chainmail -- a full exposition appears in Section \ref{sect:assertions}.

\sdparagraph{Example Configurations} We  will illustrate these features using the  \prg{Bank}/\prg{Account} example from the previous Section.
We   use the runtime configurations $\sigma_1$ and $\sigma_2$ 
shown in the left and right diagrams in Figure \ref{fig:BankAccountDiagrams}.
In both diagrams the rounded boxes depict objects:  green for those from the 
\prg{Bank}/\prg{Account} component, and grey for the ``external'',  ``client'' objects.
The transparent green rectangle  shows which objects are contained by the \prg{Bank}/\prg{Account} component.
The object at \prg{1} is a \prg{Bank}, those at \prg{2}, \prg{3} and \prg{4} are 
\prg{Account}s, and those at \prg{91}, \prg{92}, \prg{93} and \prg{94} are 
``client'' objects which belong to classes different from those from the \prg{Bank}/\prg{Account}  module.

Each configuration represents one alternative implementation of the Bank object.
Configuration  $\sigma_1$ may arise from execution using a module $M_{BA1}$, where  \prg{Account} objects
  have a field \prg{myBank} pointing to their \prg{Bank}, and an integer field  \prg{balance}
-- the code can be found in appendix~\ref{Bank:appendix} Fig.~\ref{fig:BanAccImplV1}.
Configuration  $\sigma_2$ may arise from execution using a module $M_{BA2}$,  where \prg{Account}s have a \prg{myBank}
field,  \prg{Bank} objects  have a \prg{ledger} implemented though a sequence of \prg{Node}s, each of which has a
 field pointing to an \prg{Account}, a field \prg{balance}, and a
 field \prg{next} -- the code can be found in appendix~\ref{Bank:appendix}
Figs.~\ref{fig:BanAccImplV2a} and~\ref{fig:BanAccImplV2b}.

\begin{figure}[htbp]
\begin{center}
\begin{tabular}{cc}
 \begin{minipage}{0.45\textwidth}
$\sigma_1$\\
 \includegraphics[width=\linewidth, trim=55  330 320 60,clip]{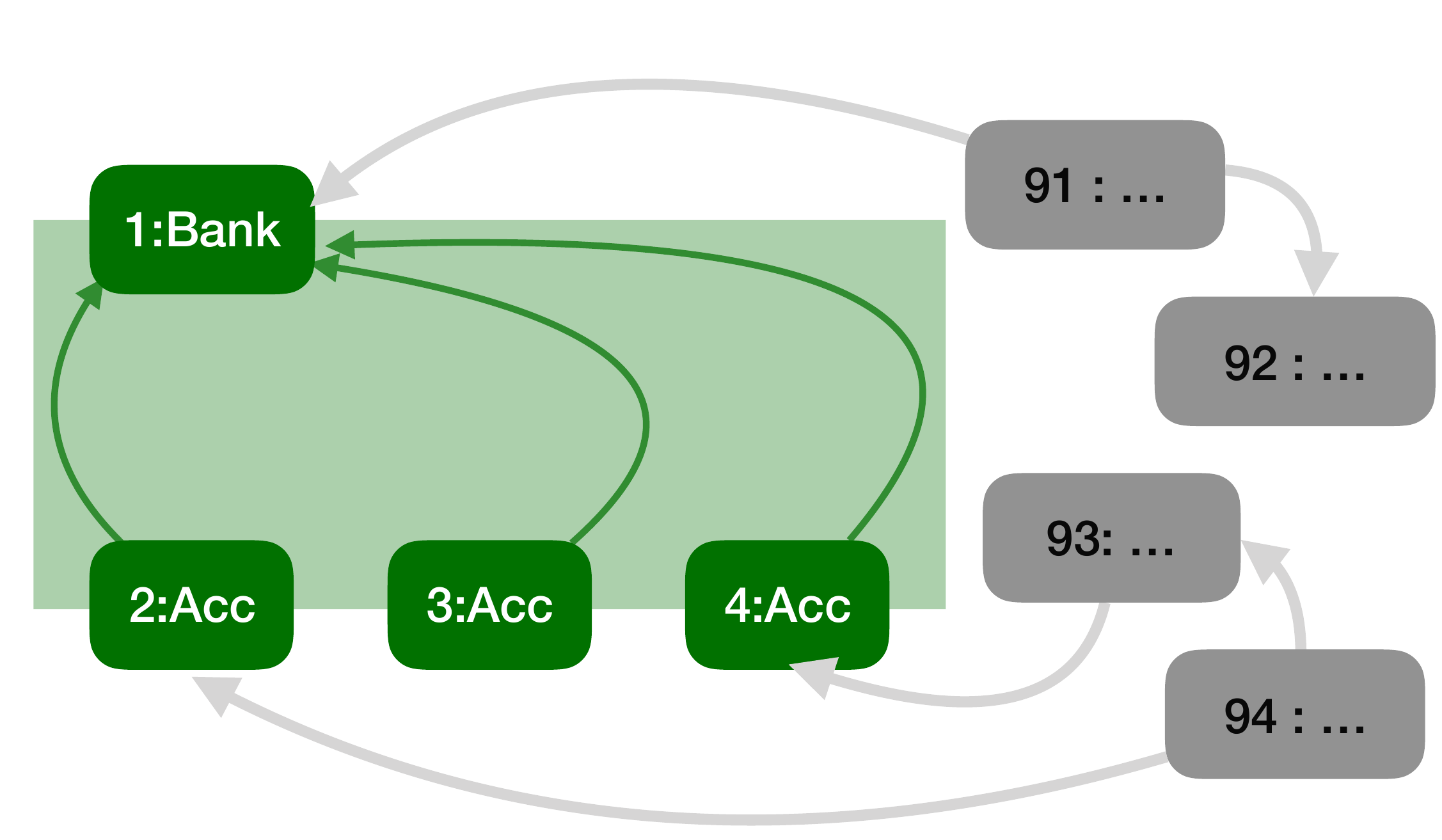}
   \end{minipage}
 &  
 \begin{minipage}{0.45\textwidth}
 $\sigma_2$\\
  \includegraphics[width=\linewidth, trim=55  330 320 60,clip]{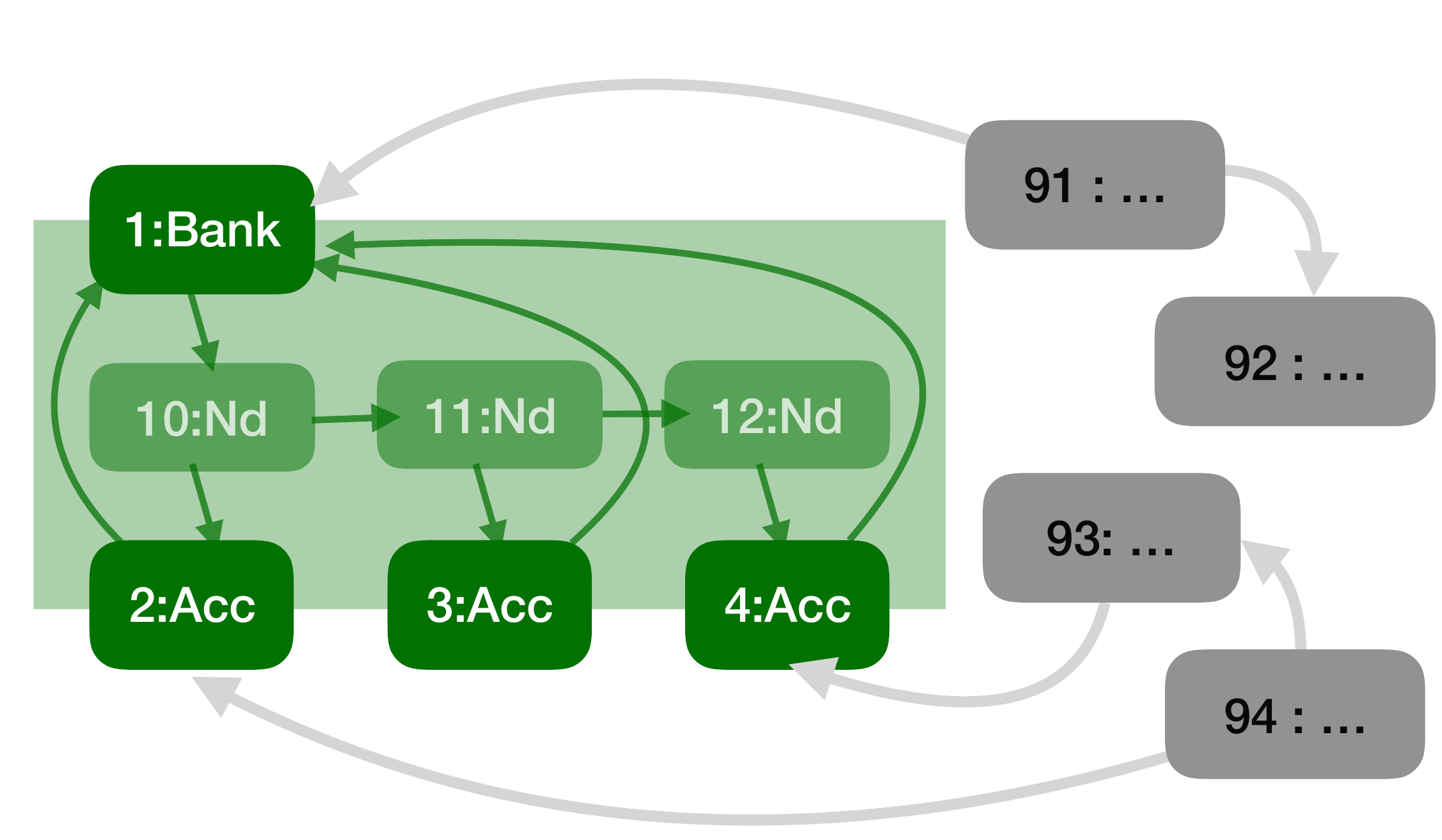}
   \end{minipage}
\end{tabular}
\end{center}
\caption{Two runtime configurations for the \prg{Bank}/\prg{Account} example. 
}
\label{fig:BankAccountDiagrams}
\end{figure}

\noindent 
For the rest,  assume variable identifiers $\prg{b}_1$, and $\prg{a}_2$--$\prg{a}_4$, and  $\prg{u}_{91}$--$\prg{u}_{94}$ denoting objects \prg{1}, \ \prg{2}--\prg{4},
and   \prg{91}--\prg{94} respectively for both $\sigma_1$ and $\sigma_2$.
That is, for $i$=$1$ or $i$=$2$, $\sigma_i(\prg{b}_1)$=\prg{1}, 
$\sigma_i(\prg{a}_2)$=\prg{2}, $\sigma_i(\prg{a}_3)$=\prg{3},  $\sigma_i(\prg{a}_4)$=\prg{4},  
  $\sigma_i(\prg{u}_{91})$=\prg{91}, $\sigma_i(\prg{u}_{92})$=\prg{92},   
  $\sigma_i(\prg{u}_{93})$=\prg{93}, and $\sigma_i(\prg{u}_{94})$=\prg{94}.

\sdparagraph{Classical Assertions} talk about the contents of the 
local variables (\ie the topmost stack frame), and the 
fields of the various objects (\ie the heap).  
  For example, the assertion  $\ \prg{a}_2.\prg{myBank}$=$\prg{a}_3.\prg{myBank} $, says that
  $\prg{a}_2$ and  $\prg{a}_3$  have the same bank. In fact, this assertion is
  satisfied in both $\sigma_1$ and $\sigma_2$, written formally as\\
  $\strut$ \hspace{1.1cm}  $...,\sigma_1 \ \models \ \prg{a}_2.\prg{myBank}=\prg{a}_3.\prg{myBank}$
  \\
 $\strut$ \hspace{1.1cm}  $...,\sigma_2 \ \models \ \prg{a}_2.\prg{myBank}=\prg{a}_3.\prg{myBank}$.

  The term \prg{x}:\prg{ClassId} says that \prg{x} is an object of class \prg{ClassId}. For example\\
  $\strut$ \hspace{1.1cm}  $...,\sigma_1 \ \models \ \prg{a}_2.\prg{myBank} : \prg{Bank}$.
  
  We support ghost fields~\cite{ghost,Leavens-etal07}, 
   \eg $\prg{a}_1$.\prg{balance} is a physical field in $\sigma_1$ and a ghost field in $\sigma_2$ since in \prg{MBA2} an \prg{Account} does not store its \prg{balance} (as can be seen in appendix~\ref{Bank:appendix}
Fig.~\ref{fig:BanAccImplV2a}). 
%
We also support the usual logical connectives, and so, we can express assertions such as \\
$\strut$ \hspace{1.1cm}    $\forall \prg{a}. [ \ \ \prg{a}:\prg{Account} \ \longrightarrow \ \ \prg{a}.\prg{myBank}:\prg{Bank}\ \wedge\  \prg{a}.\prg{balance}\geq 0\ \ ] $ .

\sdparagraph{Permission: Access}
Our first holistic assertion, $\CanAccess{\x}{\y}$, asserts that  
object $\x$ has a direct reference to another object $\y$: either one
of $\x$'s fields contains a 
reference to $\y$, or the receiver of the currently executing method is \prg{x}, and \prg{y}
is one of the arguments or a local variable. 
For example:\\
 $\strut$ \hspace{1.1cm}  $...,\sigma_1 \ \models \  \CanAccess{\prg{a}_2}{\prg{b}_1}$
\\
If  $\sigma_1$ 
were executing the method body corresponding to the call ${\prg{a}_2}$.\prg{deposit}\prg{(}${\prg{a}_3}$,\prg{360}\prg{)},  then
we would 
  have\\
 $\strut$ \hspace{1.1cm}  $...,\sigma_1 \ \models \  \CanAccess{\prg{a}_2}{\prg{a}_3}$, \\
 Namely, during execution of \prg{deposit}, the object  at   $\prg{a}_2$ has access to the object at $\prg{a}_3$, and could,
  if the method body chose to,  call a method on $\prg{a}_3$ , or  store a reference to $\prg{a}_3$ in its own fields. 
 \sophia{Access} is not symmetric, nor transitive:\\
  $\strut$ \hspace{1.1cm}  $...,\sigma_1 \ \not\models \  \CanAccess{\prg{a}_3}{\prg{a}_2}$, \hspace{0.6cm}\\
  $\strut$ \hspace{1.1cm} 
  $...,\sigma_2 \ \models \  \CanAccessTr{\prg{a}_2}{\prg{a}_3}$, \hspace{0.6cm}
 $...,\sigma_2 \ \not\models \  \CanAccess{\prg{a}_2}{\prg{a}_3}$.


\sdparagraph{Control: Calls}
The  assertion $\Calls {\x} {\y} {\m} {\zs}$
 holds 
in configurations where a method on object 
${\x}$ makes a method call ${\y}.{\m}({\zs})$ --- that is it calls method 
{\m} with object {\y} as the receiver, and with arguments {\zs}.
For example, \\
 $\strut$ \hspace{1.1cm}  $...,\sigma_3 \models \  \Calls {\x}{\prg{deposit}}  {\prg{a}_2} { {\prg{a}_3},\prg{360}}$.\\
 means that the receiver in 
 $\sigma_3$ is \x, and that
 $\prg{a}_2.\prg{deposit}{\prg{(}}\prg{a}_3,\prg{360}{\prg{)}}$
 is the next statement to be executed.


\sdparagraph{Space: {In}}
The space assertion $\Using{\A}{\SF}$ establishes validity of $\A$ 
 in a configuration  restricted to the 
objects from the set \SF.
For example, 
if  object \prg{94}  is included in $\SF_1$ but not in  $\SF_2$, then we   have\\ 
 $\strut$ \hspace{1.1cm}  $..., \sigma_1 \ \models \ \Using{ (\exists \prg{o}.\,\CanAccess{\prg{o}}{\acc_4})}{\SF_1}$
\\ 
 $\strut$ \hspace{1.1cm}  $..., \sigma_1 \ \not\models \ \Using{ (\exists \prg{o}.\,\CanAccess{\prg{o}}{\acc_4})} {\SF_2}$.\\
 The set \SF\ in the assertion $\Using{\A}{\SF}$  is therefore {\em not} the footprint of   $\A$;
  it is more like the \emph{fuel}\cite{stepindex}  given to establish that assertion. Note that  $..., \sigma \ \models \Using {\A} {\SF}$ does not imply  
  $..., \sigma \ \models \A$  nor does it imply   $..., \sigma \ \models \Using {\A} {\SF\cup\SF'}$.
  The other direction of the implication does not hold either.

\sdparagraph{Time: Next, Will, Prev, Was}
We support several operators from temporal
logic: ($\Next \A$, $\Future \A$,  $\Prev \A$, and $\Past \A$) to
talk about the future or the past in one or more \sophia{steps}.
The assertion $\Future \A$ expresses that 
$\A$ will hold in one or more steps. For example, 
taking $\sigma_4$ to be similar to  $\sigma_2$, the next statement to be executed 
to be  $\prg{a}_2.\prg{deposit}{\prg{(}}\prg{a}_3,\prg{360}{\prg{)}}$, and 
$\M_{BA2}\mkpair ..., \sigma_4 \ \models \  \acc_2.\bal=\prg{60}$,  and that
$\M_{BA2}\mkpair ..., \sigma_4 \ \models \  \acc_4.\bal\geq\prg{360}$,
then\\ 
 $\strut$ \hspace{1.1cm}  $\M_{BA2}\mkpair ..., \sigma_4 \ \models \ \Future{ {\acc_2.\bal}=\prg{420}}$.\\
The \emph{internal} module, $\M_{BA2}$ is needed for looking up the method body of \prg{deposit}.
  
\sdparagraph{Viewpoint: --  External}
The assertion $\External {\prg{x}}$ expresses that the object at {\prg{x}} does not belong to the module under consideration. 
  For example, \\
$\strut$ \hspace{1.1cm}  $\M_{AB2}\mkpair ..., \sigma_2 \ \models \ \External{\pu_{92}}$,
\hspace{1cm}  $\M_{AB2}\mkpair ..., \sigma_2 \ \not\models \ \External{\acc_2}$, \\
$\strut$
 \hspace{1.1cm}  $\M_{AB2}\mkpair ..., \sigma_2 \ \not\models \ \External{\pb_{1}.\prg{ledger}}$\\
The \emph{internal} module, $\M_{BA2}$, is needed to judge which objects are internal or external.
 
\sdparagraph{Change and Authority:} We have used 
$\Changes {...}$ 
in our \Chainmail assertions in section~\ref{sect:motivate:Bank}, as in
 $\Changes  {\acc.\prg{balance}}$. Assertions that talk about change, or give conditions for change
to happen are fundamental for security; the ability to cause change is called \emph{authority} in \cite{MillerPhD}. 
We could encode change using the other features of \Chainmail, namely, for any expression \e: 

$\strut$ \hspace{1.1cm}
$\Changes {\e}$\  \ $\equiv$\ \ $\exists v. [\ \e=v \wedge \Next {\neg ( \prg{e}=v)}\ ]$.\\
and similarly for assertions.
%

%
%
%
%

\sdparagraph{Putting these together} We now look at some composite assertions which use  
 several features from above. The assertion below 
says that if the statement to be executed is   $\acc_2.\prg{deposit}\prg{(}\acc_3,\prg{60}\prg{)}$,
then the balance of $\acc_2$ will eventually change:\\
\noindent $\M_{BA2}\mkpair ..., \sigma_2 \ \models \ {\Calls {..}   {\prg{deposit}} {\acc_2} {\acc_3,\prg{60}} } \longrightarrow {\Future{ \Changes {\acc_2.\bal}}}$.

\vspace{.2cm}
 
Now look deeper into   space assertions, $\Using{\A}{\SF}$: They allow us to characterise the set of objects which have authority over certain effects (here $\A$). In particular,  the assertion   $\Using {\Future {\A}} {\SF}$  requires two things: i) that $\A$ will hold in the future, and ii)  that the objects which cause the effect which will make $\A$ valid, are   included in \SF.
Knowing who has, and who has not, authority over properties or data is a fundamental concern of robustness
\cite{MillerPhD}. Notice that the authority is a set, rather than a single object: quite often it takes \emph{several objects in concert}
 to achieve an effect.

Consider assertions (2) and (3) from the previous section. They
  both have the form ``$\Future {\Using {\A} {\SF}}\ \longrightarrow P(\SF)$'',
  where $P$ is some property over a set. These \sophia{assertions} say that if
  ever in the future $\A$ becomes valid, and if the objects involved
  in making $\A$ valid are included in $\SF$, then $\SF$ must satisfy
  $P$. Such assertions can be used to restrict whether $\A$ will
  become valid. If we have some execution which only involves objects which do not satisfy $P$, then we know that the execution will not ever make $\A$ valid.


 
\sdparagraph{In summary,} in addition to classical 
logical connectors and classical assertions over the contents of the heap and the stack, 
our holistic assertions draw from some concepts from object capabilities
($\CanAccess{\_}{\_}$  for  permission; {$\Calls {\_} {\_} {\_} {\_}$ and  $\Changes{\_}$ for
authority) 
as well as temporal logic ($\Future \A$, $\Past \A$ and friends), and the relation of
our spatial connective ($\Using{\A}{S}$)  with ownership and effect
systems \cite{typeEffect,ownalias,ownEncaps}.

The next two sections  discuss the semantics of \Chainmail. Section \ref{sect:formal}
contains an overview of the formal model and section \ref{sect:assertions} focuses on the most important part of \Chainmail : assertions.



%% file: overModel.tex

Having outlined the ingredients of our holistic specification
language, the next question to ask is: When does a module $\M$ satisfy
a holistic assertion $\A$?  More formally: when does
$\M \models \A$ 
hold? 
  
Our answer has to reflect the fact that we are dealing with an  
\emph{open  world},  where  $\M$, our module, may be
linked with \textit{arbitrary untrusted code}.
%
%
%
%
To 
 model the open world, we consider
 pairs of modules, 
$\M \mkpair {\M'}$,  where $\M$ is the module 
whose code is supposed to satisfy the assertion,
and $\M'$  is  another 
 module which exercises
the functionality of $\M$. We call our module $\M$ the {\em internal} module, and
 $\M'$ the {\em external} module, which represents potential
 attackers or adversaries.
     
We can now answer the question: $\M \models \A$ 
 holds if for all further, {\em potentially adversarial}, modules $\M'$ and in  all runtime configurations $\sigma$ which may be observed as arising from the  execution of the code of $\M$ combined with that of $\M'$, the assertion $\A$ is satisfied. More formally, we define:\\
$~ \strut  \hspace{1.3in} \M \models \A \ \ \  \ \ \ \ \ \mbox{
if               } \ \ \  \ \ \  \  \forall \M'.\forall \sigma\in\Arising
{\M \mkpair  {\M'}}. [\ \M \mkpair  {\M'},\sigma \models \A\ ]$.  \\
Module $\M'$ represents all possible clients of {\M}.  As it is arbitrarily chosen, it reflects the open world nature of our specifications.%


The judgement $\M \mkpair  {\M'},\sigma \models \A$ means that  
assertion $\A$ is satisfied by  $\M \mkpair  {\M'}$ and $\sigma$.  
As in traditional specification languages \cite{Leavens-etal07,Meyer92}, satisfaction is judged 
in the context of a runtime configuration $\sigma$; but in addition, it is judged in the context of the internal and external modules.
These are used to find   abstract functions defining ghost fields as well as  method bodies
needed when judging validity of temporal assertions such as
$\Future {\_}$.} 

We distinguish between internal and external modules. This 
\sophia{has two uses:}
First, 
\Chainmail\ includes the ``$\External{\prg{o}}$'' assertion to require
that an object belongs to the external module, as in the Bank
Account's assertion (2) and (3) in
section~\ref{sect:motivate:Bank}. Second, we adopt a version of
visible states semantics \cite{MuellerPoetzsch-HeffterLeavens06,larch93,Meyer97}, treating all
executions within a module as atomic.
We only record runtime configurations which are {\em external}
 to module $\M$, \ie those where the
 executing object (\ie the current receiver) comes from module $\M'$.
 Execution 
 has the form\\
 $~ \strut  \hspace{1.3in}    \M \mkpair  {\M'},\sigma \leadsto \sigma'$\\  
where we ignore all intermediate steps
 with receivers  internal to $\M$. 
%
In the next section we  shall 
outline the underlying programming language, and
define the judgment  $\M \mkpair  {\M'},\sigma \leadsto \sigma'$ and the set 
$\Arising {\M \mkpair  {\M'}}$.

%% file: summaryExecution.tex
\renewcommand{\appref}[1]{, c.f. Appendix, Def.\,\ref{#1}}
 
The meaning of \Chainmail assertions is parametric with an
underlying object-oriented programming language, with modules  as repositories of code, classes with fields, methods and
ghostfields, objects described by classes, a way to link  modules into larger ones, and a concept of 
program execution.\footnote{We believe that \Chainmail can be applied to 
any language with these features.}

We have developed   \LangOO, a minimal such object-oriented language, which we
outline in  this section. 
We  describe the novel aspects of \LangOO,, and 
summarise the more conventional parts, relegating  full, and mostly unsurprising,
definitions 
to Appendix \ref{app:LangOO},

Modules are central to \LangOO, as they are to \Chainmail. As modules are repositories
of code, we adopt the common formalisation of modules as maps from 
class identifiers to class definitions\appref{defONE}. We use the terms module and component in an
analogous manner to class and object respectively.  \LangOO is untyped 
-- \sophia{this has several reasons: Many popular programming languages are untyped.
The external module might be untyped, and so it is more
general to consider everything as untyped.
Finally, 
a solution that works for an untyped language 
will also apply to a typed language;' the converse is not true.
}

 Class
definitions consist of field, method and ghost field declarations\appref{def:syntax:classes}.
Method bodies are sequences of 
statements, which  can be field read or field assignments, object
creation, method calls, and return statements. 
Fields are private in the sense of C++: they can only be read or
written by methods of the current class.
This is enforced by the operational semantics, \cf Fig.  \ref{fig:Execution}.
We  discuss ghost fields in the next section.

Runtime configurations, $\sigma$,  contain   all the usual information about execution snapshots: the heap, and a
stack of frames. 
%
Each frame consists of a continuation, \prg{contn}, describing the remaining code to be executed by the
frame, and a map from
variables to values. Values are either addresses or sets of addresses; the latter 
are needed to deal with assertions which quantify over sets of objects, as
\eg (1) and (2) from section \ref{sect:motivate:Bank}.
We define {\emph{one-module} execution  through a judgment of the form $\M, \sigma \leadsto \sigma'$ in the Appendix, Fig.  \ref{fig:Execution}. 

We define a module linking operator \  $\link$ \  so that
$\M\link\M'$ is the union of the two modules, provided that their domains are disjoint\appref{def:link}.
As we said in section \ref{sect:overviewmodel}, we distinguish  between the internal and external module. \sophia{We consider execution  from the view of the
external module, and treat}  execution of 
methods from the internal module as atomic. For this, we define \emph{two-module execution}  based on
one-module execution as follows:

\begin{definition}
\label{def:execution:internal:external}
\label{def:module_pair_execution} 
Given runtime configurations $\sigma$,  $\sigma'$,  and a module-pair $\M \mkpair \M'$ we define
execution where $\M$ is the internal, and $\M'$ is the external module as below:
 
\begin{itemize}
\item
$\M \mkpair \M', \sigma \leadsto \sigma'$ \IFF
there exist  $n\geq 2$ and runtime configurations $\sigma_1$,  ...
$\sigma_n$, such that
\begin{itemize}
\item
$\sigma$=$\sigma_1$,\ \  \ \ and\ \ \ \ $\sigma_n=\sigma'$.
\item
$\M \link \M', \sigma_i \leadsto \sigma_{i+1}'$,\  \  for $1\leq i \leq n\!-\!1$
\item
$\ClassOf{\this} {\sigma}\not\in dom({\M})$,  \ \  \ \ and\ \ \ \
$\ClassOf{\this} {\sigma'} \not\in dom({\M})$,
\item
 $\ClassOf{\this} {\sigma_i} \in dom({\M})$,\ \ \ \ for $2\leq i \leq n\!-\!2$
\end{itemize}
\end{itemize}

\end{definition}
 
In the definition above,  $\ClassOf {\x} {\sigma} $ looks up the class of the object \sophia{stored} at \x\appref{def:interp}.
 For example, for $\sigma_4$ as in Section \ref{sect:chainmail} whose next statement to be executed 
 is  $\prg{a}_2.\prg{deposit}{\prg{(}}\prg{a}_3,\prg{360}{\prg{)}}$,  we would have 
 a sequence of configurations $\sigma_{41}$, ... $\sigma_{4n}$,  $\sigma_{5}$ so that the
  one-module execution gives
 $\M_{BA2}, \sigma_4 \leadsto \sigma_{41} \leadsto \sigma_{42} ... \leadsto \sigma_{4n}   \leadsto \sigma_{5}$.
This would correspond to an atomic evaluation in the two-module execution: \  \
 $\M_{BA2}\mkpair \M', \sigma_4 \ \leadsto \sigma_5$ (see Fig.\ref{fig:VisibleStates}; \sophia{where blue stands for $\sigma(this)\!\in\!M_1$,%
   and  orange for $\sigma(this)\!\in\!M_2$}).

\begin{figure}[htb]
  \vspace*{-2.5mm}
  \begin{center}
   \begin{minipage}{0.80\textwidth}
     \begin{center}
       \includegraphics[width=\linewidth]{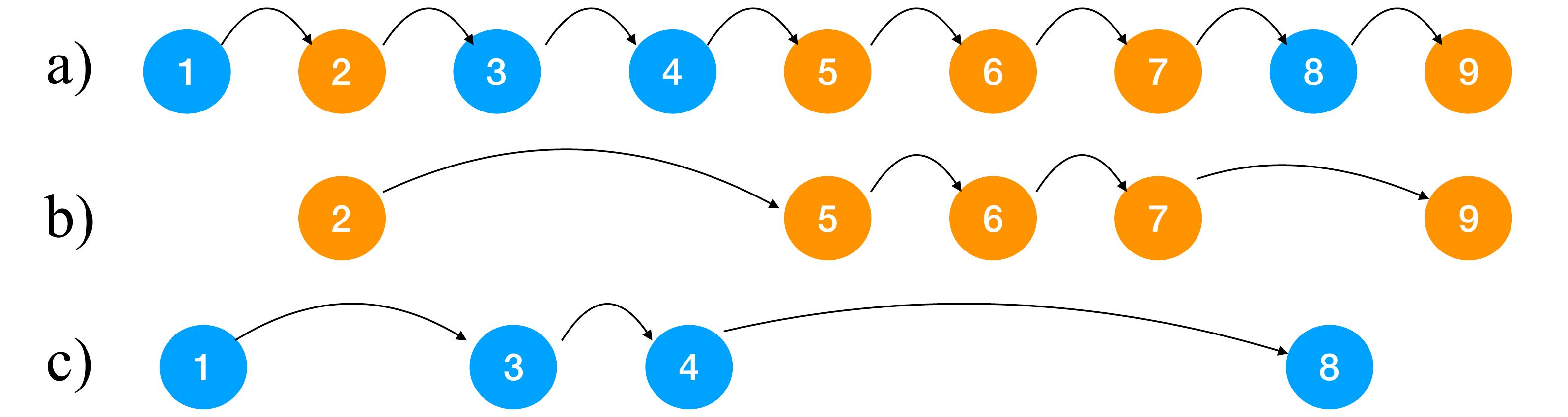}
     \end{center}
   \end{minipage}
   \end{center}
   \vspace*{-2.5mm}
   \caption{Two Module Execution
     (Def. \ref{def:execution:internal:external}). %
     a) $\M_1 \link \M_2$ b) $\M_1 \mkpair \M_2$ c) $\M_2 \mkpair \M_1$}
   \label{fig:VisibleStates}
 \end{figure}

Two-module execution  is related to 
visible states
semantics \cite{MuellerPoetzsch-HeffterLeavens06} as they both filter configurations, with  the difference
that in visible states semantics 
execution is unfiltered and  configurations are only filtered when it comes to the  consideration 
of class invariants} while {two-module execution filters  execution.
The lemma below says  that linking is associative and commutative, and preserves   both one-module and two-module execution.

\begin{lemma}[Properties of linking]
\label{lemma:linking}
 For any modules $\M$,   $\M'$, $\M''$, and $\M'''$ and runtime configurations $\sigma$, and $\sigma'$ we have$:$
 \label{lemma:linking:properties}

 \begin{itemize}
     \item
     $(\M \link \M')\link \M''$ = $\M \link (\M' \link \M'')$  \hspace{1cm} and    \hspace{1cm}   $\M \link \M'$  = $\M' \link\M$.
      \item
      $\M, \sigma \leadsto \sigma'$, and $\M\link \M'$ is defined, \  \ \ \ \  implies\ \ \ \ \   $\M\link \M', \sigma \leadsto \sigma'$.
 \item
 $\M \mkpair \M', \sigma \leadsto \sigma'$   \  \ \ \ \  implies\ \ \ \ \  $(\M\link\M'') \mkpair (\M'\link\M''') ,\sigma \leadsto \sigma'$.  
  \end{itemize}

 \end{lemma}
 
 We can now answer the question as to which runtime configurations are pertinent when judging a module's
adherence to an assertion.
\sophia{{\em Initial configurations}, are those whose heap  have only one object, of class \prg{Object}, and whose stack} have one frame, with arbitrary continuation.
\sophia{{\em Arising}} configurations are  those that can be reached by two-module execution, starting from any initial configuration.
 
\begin{definition}[Initial and Arising Configurations] are defined as follows: \label{def:arise}

\begin{itemize}
     \item
   $\Initial {(\psi,\chi)}$, \ \ if \ \ $\psi$ consists of a single frame $\phi$ with $dom(\phi)=\{ \this \}$, and there exists  some address $\alpha$, such that \ \ \    $\interp {\this}{\phi}$=$\alpha$, and \ $dom(\chi)$=$\alpha$,\  and\  
    $\chi(\alpha)=(\prg{Object},\emptyset)$.
 \item
 $\Arising  {\M\mkpair\M'} \ = \ \{ \ \sigma \ \mid \ \exists \sigma_0. \ [\  \Initial{\sigma_0} \  \ \wedge\ \  \M\mkpair\M', \sigma_0 \leadsto^* \sigma \ \ ] \ \ \} $
 \end{itemize}

\end{definition}

%% file: assertions.tex
\Chainmail assertions (details in appendix \ref{sect:standard}) consist of (pure) expressions \e, comparisons between expressions, classical
assertions about the contents of heap and stack, the usual logical
connectives, as well as our holistic concepts.
In this section we focus on the novel,
 holistic, features of \Chainmail (permission, control, time, space, and viewpoint),
as well as our wish to support some form of recursion while keeping the logic of assertions classical.

\subsection{Satisfaction of Assertions - Access, Control, Space, Viewpoint}
\label{sect:pcsv} 

\textit{Permission} expresses that an object has the potential to call
methods on another object, and to do so directly, without  help from any
intermediary object. This is the case when the two objects are aliases, 
or the first object has a field pointing to the second object, or
the first object is the receiver of the currently executing method and the second object is one of the 
arguments or a local variable. Interpretations of variables and paths, $\interp {...} {\sigma}$, are defined
in the usual way (appendix Def.~\ref{def:interp}).

\begin{definition}[
Permission]  \label{def:valid:assertion:access}
For any modules $\M$, $\M'$,  variables \prg{x} and \prg{y}, we define
\begin{itemize}
\item
$\M\mkpair \M', \sigma \models  \CanAccess{\prg{x}}{\prg{y}}$   \IFF  \sd{$\interp {\x} {\sigma}$ and $\interp {\y} {\sigma}$ are defined}, and \begin{itemize}
\item
$\interp {\x} {\sigma}$=$\interp {\y} {\sigma}$, \ \ or
\item
$\interp {\x.\f} {\sigma}$=$\interp {\y} {\sigma}$, \  \ for some field \prg{f},  \ \ or
\item
$\interp {\x} {\sigma}$=$\interp {\this} {\sigma}$ and
  $\interp {\y} {\sigma}$=$\interp {\z} {\sigma}$,\ \ \ 
for some variable \z\ and  \z\ appears in  $\sigma$.\prg{contn}.
 \end{itemize}
\end{itemize}
\end{definition}

\noindent 
In the last disjunct, where \z\ is a parameter or local variable,
we  ask that   \z\ appears in the code being executed ($\sigma$.\prg{contn}).
This requirement 
ensures that variables which were introduced into the variable map 
in order to give meaning to existentially quantified assertions, are not considered.

\vspace{.2cm} \noindent
\textit{Control} expresses which object is the process of making a function call on another object and
with what arguments. The relevant information
\sophia{is} stored in the continuation (\prg{cont}) on the top frame.
\begin{definition}[
Control]  \label{def:valid:assertion:control}
For any modules $\M$, $\M'$,  variables \prg{x} , \y, $\prg{z}_1,...\prg{z}_n$, we define$:$
\begin{itemize}
   \item
$\M\mkpair \M', \sigma \models  \Calls {\prg{x}}  {\prg{m}} {\prg{y}}  {\z_1,...\z_n}$ \IFF \ \ \sd{$\interp {\x} {\sigma}$, $\interp {\y} {\sigma}$, $\interp {\z_1} {\sigma}$, ... $\interp {\z_n} {\sigma}$ are defined},\ \  and 
\begin{itemize}
\item
$\interp{\prg{this}}{\sigma}$=$\interp{\prg{x}}{\sigma}$, \ and
\item
$\sigma.\prg{contn}$=${\uu.\m(\v_1,..\v_n);\_}$,\ \ \ for some  $\uu$,$\v_1$,... $\v_n$, \ and
\item
 $\interp{\prg{y}}{\sigma}$=$\interp{\prg{u}}{\sigma}$,\ \ \ and \ \ \ 
  $\interp{\z_i}{\sigma}$=$\interp{{\prg{v}_i}}{\sigma}$, for all  $i$.
 \end{itemize}
  \end{itemize}
\end{definition}
\noindent 
Thus,  $\Calls {\prg{x}}  {\prg{m}} {\prg{y}}  {{\z_1,...\z_n}}$ expresses the  call 
$\y.\m(\z_1,...\z_n)$ will be executed next, and that the caller is \x.

\vspace{.2cm} \noindent
\sd{\textit{Viewpoint} is about whether an object is viewed as belonging to   the internal mode;
this is determined by the class of the object.}

 \begin{definition}[
Viewpoint]  \label{def:valid:assertion:view}
For any modules $\M$, $\M'$, and variable\x, we define
\begin{itemize}
 \item
$\M\mkpair \M', \sigma \models \External {\x}$ 
  \IFF
 \sd{$\interp {\x} {\sigma}$ is defined} and $\ClassOf {\interp{\x}{\sigma}} {\sigma} \notin dom(\M)$
\item
$\M\mkpair \M', \sigma \models \Internal {\x}$ 
  \IFF
  \sd{$\interp {\x} {\sigma}$ is defined} and $\ClassOf {\interp{\x}{\sigma}} {\sigma} \in dom(\M)$
\end{itemize}
\end{definition}
\noindent

\vspace{.2cm} 

\noindent
\sd{\textit{Space} is about  asserting that some property \A\ holds in a configuration whose objects are restricted to those
from a give set \SF. This way we can express that the objects from the set \SF\ have authority over the assertion \A.}
\sd{In order to define validity of $\Using {\A} {\prg{S}}$ in a configuration $\sigma$, 
we first define a restriction operation,  $\restrct {\sigma}{\prg{S}}$ which restricts the objects from $\sigma$ to only those
from $\SF$.
}

 \begin{definition}[Restriction of Runtime Configurations]  \label{def:restrict}
The restriction operator~$\;\restrct{} {} $ applied to a runtime configuration $\sigma$ and a variable  \prg{S} is defined as follows:
 \label{def:config:restrct}
 $~ $
\begin{itemize}
\item
$\restrct {\sigma}{\prg{S}}  \triangleq  (\sd{\psi}, \chi')$, \ if \  $\sigma$=$(\psi,\chi)$, \    $dom(\chi')=\interp {\prg{S}} {\sigma}$, and   
 $\forall \alpha\!\in\!dom(\chi').\chi(\alpha)=\chi'(\alpha)$.
\end{itemize}
\end{definition}

\begin{tabular}{cc}
 \begin{minipage}{0.45\textwidth}
\sd{For example, if we take $\sigma_2$ from Fig. \ref{fig:BankAccountDiagrams} in Section \ref{sect:motivate:Bank},
and restrict it with some set $\SF_4$ such that $\interp {\SF_4}{\sigma_2}=\{ 91, 1, 2, 3, 4, 11 \}$,
then the restriction $\restrct {\sigma_2}{\prg{S}_4}$ will look as on the right.}
  \end{minipage}
 &  
 \begin{minipage}{0.45\textwidth}
  \includegraphics[width=\linewidth, trim=55  330 320 60,clip]{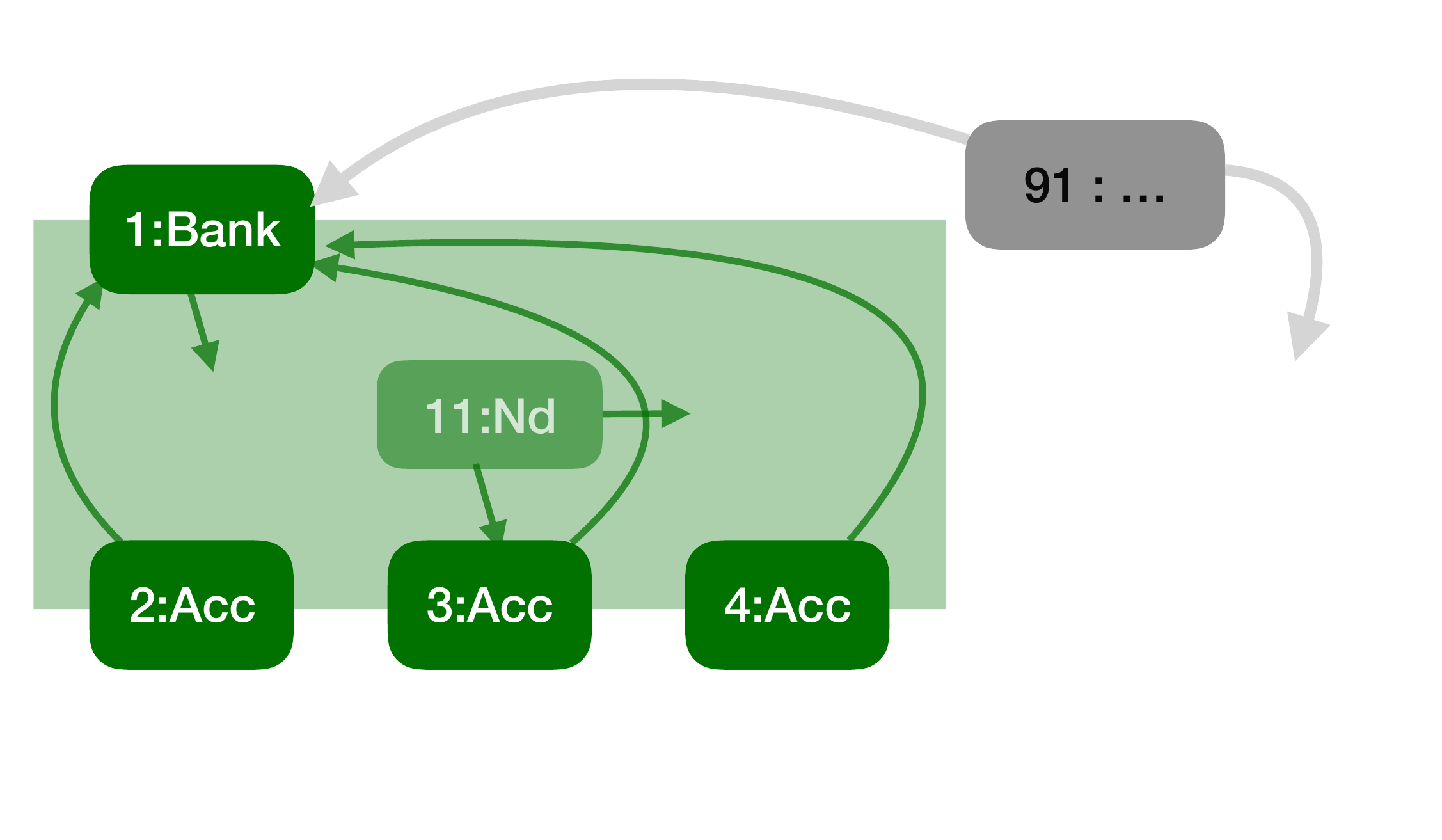}
   \end{minipage}
\end{tabular}
 
\sd{Note in the diagram above the dangling pointers at objects $1$,  $11$, and $91$ - reminiscent of the separation
of heaps into disjoint subheaps, as provided by the 
$*$ operator in separation logic \cite{Reynolds02}. The difference is, that in separation logic, the 
separation is provided through the assertions, where $\A * \A'$ holds in any heap  which can be split into
disjoint $\chi$ and $\chi'$ where $\chi$ satisfies $\A$ and $\chi'$ satisfies $\A'$. That is, in $\A * \A'$  the split of the 
heap is determined by the assertions   $\A$ and $\A'$ and there is an implicit requirement of disjointness, while in $\restrct {\sigma}{\prg{S}}$ the split is
determined by \SF, and no disjointness is required.
}

We now define the semantics of $\Using {\A} {\prg{S}}$.

\begin{definition}[Space]  \label{def:valid:assertion:space} 
For any modules $\M$, $\M'$, assertions $\A$ and variable \prg{S}, we define$:$
\begin{itemize}
\item
 $\M\mkpair \M', \sigma \models \Using {\A} {\prg{S}}$
 \IFF
 $\M\mkpair \M', \restrct \sigma {\prg{S}} \models  \A  $.
\end{itemize}
\end{definition}

\sd{
The set \SF\ in the assertion $\Using {\A} {\prg{S}}$ is related to  framing  from implicit dynamic frames \cite{IDF}:
 in an implicit dynamic frames assertion  \ $\textbf{acc}\, \prg{x.f}  * \A$, the frame $\prg{x.f}$ prescrives which locations
 may be used to determine validity of $\A$.
The difference is that frames are sets of locations (pairs of address and field), while our \prg{S}-es are sets of addresses.
 More  importantly,   implicit dynamic frames assertions whose frames are not large enough are badly formed, 
 while in our work, such assertions
are allowed and may   hold 
or not, \eg   
 $\M_{BA2}\mkpair \M', \sigma  \models \neg\ \Using {(\exists n.\acc_2.\prg{balance}=n)} {\prg{S}_4}$.
}

\subsection{Satisfaction of Assertions - Time}
\label{sect:time} 
\sd{To deal with time, we are faced with four challenges:  a) validity of assertions in the future or the past needs to be judged in the 
future configuration, but using the bindings from the current one, 
b) the current configuration needs to store the code being executed, so
as to be able to calculate future configurations, c) when considering the future, we do not want to observe 
configurations which go beyond the frame currently at the top of the stack, d) there is no "undo" operator to deterministically enumerate
all the previous configurations.}

\sd{\sophia{Consider  challenge  a)  in some more detail: the assertion
$\Future {\x.\f=\prg{3}}$  is satisfied in the \emph{current} configuration,   $\sigma_1$,
 if in some {\em future} configuration, $\sigma_2$, the field  \f\, of the object that is pointed at 
 by \x\, in the {\em current} configuration ($\sigma_1$) has the value \prg{3}, that is, if
$\interp{ \interp{\x}{\sigma_1}.f}{\sigma_2} = 3$,   even if
in that future configuration \x\ denotes a different object (i.e. if $\interp{\x}{\sigma_1}\neq \interp{\x}{\sigma_2}$).}}
\sd{
To address this, we define an  auxiliary concept:
the operator\ $\adapt$\,,   where}
 $\sigma_1 \adapt\, \sigma_2$ adapts the second configuration to the top frame's view of the former: 
 it returns a new configuration whose stack \sophia{comes from $\sigma_2$  but is augmented with the view from the  top frame} 
   from $\sigma_1$ and where the continuation   has been consistently renamed,. This allows us to interpret expressions  in   
   $\sigma_2$ but with the variables bound according to   $\sigma_1$; \eg we can obtain that value 
   of \prg{x} in configuration  $\sigma_2$ even if \prg{x} was out of scope in $\sigma_2$.  

 \begin{definition}[Adaptation] 
  \label{def:config:adapt}
 For runtime configurations $\sigma_1$, $\sigma_2$.$:$
 $~ $ 
\sophia{
\begin{itemize}
\item
$\sigma_1 \adapt \sigma_2 \triangleq (\phi_3\cdot\psi_2,\chi_2)$  \IFF 
\begin{itemize}
\item
  $ \phi_3=(\, \prg{contn}_2[\prg{zs}_2  / \prg{zs}' ], 
  \varMap_2[\prg{zs}'\mapsto \varMap_2({\prg{zs}}_2)][\prg{zs}_1\mapsto \varMap_1({\prg{zs}_1})] \, ) $,\ \ \ 
 where
\item
$\sigma_1=(\phi_1\cdot\_,\_)$,\ \ \  $\sigma_2= (\phi_2\cdot\psi_2,\chi_2)$, \ \ \  
$\phi_1$=$(\_,\varMap_1)$, \ \ \  $\phi_2$=$(\prg{contn}_2,\varMap_2)$, \ \ 
 and
\item
$\prg{zs}_1$=$dom(\varMap_1)$,  \ \  \ $\prg{zs}_2$=$dom(\varMap_2)$,  \ \  \ and \ \ \
\item
$\prg{zs}'$ is a set  of variables with  the  same cardinality as $\prg{zs}_2$, and
all variables in
$\prg{zs}'$  are fresh in $\varMap_1$ and in $\varMap_2$.
\end{itemize}
\end{itemize}
}
\end{definition}

\sophia{That is, in the new frame $\phi_2$ from above, we keep the same continuation as from $\sigma_2$ but rename all
variables with fresh names $\prg{zs}'$,   and combine the variable map  $\beta_1$  from $\sigma_1$ 
with the variable map $\beta_2$ from $\sigma_2$  while avoiding names clashes through
 the renaming $[\prg{zs}'\mapsto \varMap_2(\prg{zs}_2)]$.}
The consistent renaming of the continuation allows the correct modelling of execution (as needed,   for the semantics of  nested time assertions, as \eg in $\Future {\x.\f=\prg{3} \wedge \Future {\x.\f=\prg{5}}}$).

\sd{
Having addressed challenge a) we turn our attention to the remaining challenges: We address challenge b) by storing the remaining code to be executed in \prg{cntn} in each frame. 
We address challenge c) by only taking the top of the frame when considering future executions.
Finally, we address challenge d) by considering only configurations which arise from initial configurations, and 
which lead to the current configuration. 
}

\begin{definition}[Time Assertions]  \label{def:valid:assertion:time}
For any modules $\M$, $\M'$, and assertion  $\A$ we define
  
\begin{itemize}
 \item
  $\M\mkpair \M', \sigma \models  \Next \A $
  \IFF
  $\exists \sigma'.\, [\ \ \M\mkpair \M',\phi \leadsto  \sigma' \ \wedge \M\mkpair \M',\sigma\adapt\sigma' \models \A \ \  ]$,
 \\
$\strut ~ \hspace{1.4in} $ \hfill  and where $\phi$ is
so that $\sigma$=$(\phi\cdot\_,\_)$.\item
  $\M\mkpair \M', \sigma \models  \Future \A $
  \IFF
  $\exists \sigma'.\, [\ \ \M\mkpair \M',\phi \leadsto^* \sigma' \ \wedge \M\mkpair \M',\sigma\adapt\sigma' \models \A \ \  ]$,
 \\
$\strut ~ \hspace{1.4in} $   \hfill   and where $\phi$ is
so that $\sigma$=$(\phi\cdot\_,\_)$.  
  \item
 $\M\mkpair \M', \sigma \models  \Prev \A $ \IFF
 $\forall \sigma_1, \sigma_2. [\ \ \Initial{\sigma_1}\ \wedge \   \M\mkpair \M', \sigma_1  \leadsto^*  \sigma_2 $\\
 $\strut ~ \hspace{2.1in}   \wedge \   \M\mkpair \M', \sigma_2  \leadsto   \sigma  
 \ \  \ \longrightarrow \ \ \   \
 \M\mkpair \M', \sigma\adapt\sigma_2  \models \A\ \
 ]$ 
 \item
 $\M\mkpair \M', \sigma \models  \Past \A $ \IFF
%
$\forall \sigma_1. [\ \ \Initial{\sigma_1}\ \wedge \  \M\mkpair \M',\sigma_1 \leadsto^* \sigma \longrightarrow $\\
 $\strut ~ \hspace{0.1in} $   \hfill   $(\ \  \exists \sigma_2.
 \M\mkpair \M',\sigma_1 \leadsto^* \sigma_2 \ \wedge\  \M\mkpair \M',\sigma_2 \leadsto^* \sigma \ \wedge \ 
 \M\mkpair \M', \sigma\adapt\sigma_2  \models \A\ \ 
 )]$ 
\end{itemize}
\end{definition}

%
 
 In general, $\Future {\Using {\A} {\SF}}$ is different from
  $\Using {\Future {\A}} {\SF}$.  Namely, in the former assertion, $\SF$ must contain
   the objects involved in reaching the future configuration as well as the objects needed to
    then establish validity of $\A$ in that future configuration. In the latter assertion, 
     $\SF$ need only contain the objects needed to establish $\A$ in that future configuration.
  For example, revisit Fig. \ref{fig:BankAccountDiagrams}, and take $\SF_1$ to consist of objects \prg{1}, \prg{2},   \prg{4}, \prg{93}, and \prg{94},
  and $\SF_2$ to consist of objects \prg{1}, \prg{2},   \prg{4}.  Assume that 
   $\sigma_5$ is like $\sigma_1$, that the next call in $\sigma_5$ is a method on $\prg{u}_{94}$, whose  body obtains the
  address of $\acc_4$ (by making a call on \prg{93} to which it has access), and the address of $\acc_2$ (to which it has access),
  and then makes the call $\acc_2.\prg{deposit}(\acc_4,360)$. Assume also     that $\prg{a}_4$'s balance is \prg{380}.
  Then\\
  $\strut$ \hspace{1.1cm}  $\M_{BA1}\mkpair ..., \sigma_5 \ \models \ \Using {\Future{ \Changes {\acc_2.\bal}}} {\SF_1}$\\
   $\strut$ \hspace{1.1cm}  $\M_{BA1}\mkpair ..., \sigma_5 \ \not\models \ \Using {\Future{ \Changes {\acc_2.\bal}}} {\SF_2}$\\
 $\strut$ \hspace{1.1cm}  $\M_{BA1}\mkpair ..., \sigma_5 \ \models \ \Future{ \Using {\Changes {\acc_2.\bal}} {\SF_2}}$\

\subsection{Properties of Assertions}

\label{sect:classical} 
We define equivalence of   assertions in the usual way: assertions $\A$ and $\A'$are equivalent if they are satisfied  in
the context of the same configurations and module pairs -- \ie\\
 \strut \hspace{1.1cm} $\A \equiv \A'\  \IFF\    \forall \sigma.\, \forall \M, \M'. \ [\ \ \M\mkpair \M', \sigma \models \A\ \mbox{ if and only if }\ \M\mkpair \M', \sigma \models \A'\ \ ].$\\
We can then prove that the usual equivalences hold, \eg\  $ \A \vee \A' \ \equiv \  \A' \vee \A$, and\   $\neg (\exists \prg{x}.\A )  \  \ \equiv \  \forall \prg{x}.(\neg  \A)$.
Our assertions are classical, \eg  $ \A \wedge\neg \A \ \equiv \  \prg{false}$, and $\M\mkpair \M', \sigma  \models \A$ and  $\M\mkpair \M', \sigma  \models \A \rightarrow \A'$  implies
$\M\mkpair \M', \sigma  \models \A '$. 
 \sd{This desirable property comes at the loss of some expected equivalences, \eg, in general, 
 $\e = \prg{false}$ and $\neg\e$ are not equivalent. 
 More  in Appendix \ref{app:assertions}.}

\subsection{Modules satisfying assertions}

Finally, we define satisfaction of assertions by modules: A module $\M$ satisfies an assertion $\A$ if for all modules $\M'$, in all configurations arising from executions of $\M\mkpair\M'$, the assertion $\A$ holds.

\begin{definition}
\label{def:module_satisfies}
For any module $\M$, and  assertion $\A$, we define:
\begin{itemize}
\item
$\M \models \A$ \IFF  $\forall \M'.\, \forall \sigma\!\in\!\Arising{\M\mkpair\M'}.\   \M\mkpair\M', \sigma \models \A$
\end{itemize}
\end{definition}

%% file: example-driven-design.tex


\paragraph{Examplars}

The design of \Chainmail was guided by the study of a sequence of
exemplars taken from the object-capability literature and the smart
contracts world:

\begin{enumerate}
\item \textbf{Bank} \cite{arnd18} - Bank and Account as in
Section~\ref{sect:motivate:Bank} with two different implementations.
\item
\textbf{ERC20} \cite{ERC20} - Ethereum-based token contract.
\item
\textbf{DAO} \cite{Dao,DaoBug} - Ethereum contract for Decentralised Autonomous
Organisation.
\item
\textbf{DOM} \cite{dd,ddd} - Restricting access to browser Domain Object Model\\
\end{enumerate}

\noindent
\sophia{We} present these exemplars as 
appendices \cite{examples}. Our design was also driven by work on other
examples such as the membrane \cite{membranesJavascript},
the Mint/Purse \cite{MillerPhD}, and 
Escrow \cite{proxiesECOOP2013,swapsies}.


  

\paragraph{Model}

We have constructed a Coq model \cite{coq} of the core of the Chainmail
specification language, along with the underlying \LangOO language.
Our formalism is organised as follows:
\begin{enumerate}
\item
The \LangOO Language: a class based, object oriented language with mutable references.
\item
Chainmail: The full assertion syntax and semantics defined in Definitions \ref{def:execution:internal:external}, \ref{def:arise}, \ref{def:valid:assertion:access}, \ref{def:valid:assertion:control}, \ref{def:valid:assertion:view}, \ref{def:restrict}, \ref{def:valid:assertion:space}, \ref{def:config:adapt}, \ref{def:valid:assertion:time} and \ref{def:module_satisfies}.
\item
\LangOO Properties: Secondary properties of the loo language that aid in reasoning about its semantics.
\item
Chainmail Properties: The core properties defined on the semantics of Chainmail.
\end{enumerate}

\sophia{In the associated appendix} (see Appendix \ref{sect:coq}) we list and present the properties of Chainmail we have formalised in Coq.
We have proven that Chainmail obeys much of the properties of classical logic. While we formalise most of the underlying semantics, we make several assumptions in our Coq formalism: (i) the law of the excluded middle,  a property that is well known to be unprovable in constructive logics, and (ii) the equality of variable maps and heaps down to renaming. Coq formalisms often require fairly verbose definitions and proofs of properties involving variable substitution and renaming, and assuming equality down to renaming saves much effort.

More details of the formal foundations of \Chainmail, and the model,
are also in appendices \cite{examples}.

%% file: related.tex
\paragraph{Behavioural Specification Languages} 

Hatcliff et al.\ \cite{behavSurvey2012} provide an excellent survey of
contemporary specification approaches.  With a lineage back to Hoare
logic \cite{Hoare69}, Meyer's Design by Contract \cite{Meyer97} was the
first popular attempt to bring verification techniques to
object-oriented programs as a ``whole cloth'' language design in
Eiffel.  Several more recent specification languages are now making
their way into practical and educational use, including JML
\cite{Leavens-etal07}, Spec$\sharp$ \cite{BarLeiSch05}, Dafny
\cite{dafny} and Whiley \cite{whiley15}. Our approach builds upon
these fundamentals, particularly Leino \& Shulte's
formulation of
two-state invariants \cite{usingHistory}, and Summers and
Drossopoulou's Considerate Reasoning \cite{Considerate}.
In general, these approaches assume a closed system, where modules
can be trusted to cooperate. In this paper we aim to
work
in an open system where modules'
invariants must be protected irrespective of the behaviour of the rest
of the system.

\paragraph{Defensive Consistency}


In an open world, we cannot rely on the kindness of strangers: rather
we have to ensure our code is correct regardless of whether it
interacts with friends or foes.
Attackers 
\textit{``only have to be lucky once''} while secure systems 
\textit{``have to be lucky always''} \cite{IRAThatcher}.
Miller \cite{miller-esop2013,MillerPhD} defines the necessary approach
as \textbf{defensive consistency}: \textit{``An object is defensively
  consistent when it can defend its own invariants and provide correct
  service to its well behaved clients, despite arbitrary or malicious
  misbehaviour by its other clients.''}  Defensively consistent
modules are particularly hard to design, to write, to understand, and
to verify: but
they make it much
easier to make guarantees about systems composed of multiple components
\cite{Murray10dphil}.

\paragraph{Object Capabilities and Sandboxes.}
{{\em Capabilities} as a means to support the development of concurrent and distributed system  were developed in the 60's
by Dennis and Van Horn \cite{Dennis66}, and were adapted to the
programming languages setting in the 70's \cite{JamesMorris}. 
{\em Object capabilities} were first introduced~\cite{MillerPhD} in the early 2000s},
 and many recent 
studies manage
to verify  safety or correctness of object capability programs.
Google's Caja \cite{Caja} applies   sandboxes, proxies, and wrappers
 to limit components'
access to \textit{ambient} authority.
Sandboxing has been validated
formally: Maffeis et al.\ \cite{mmt-oakland10} develop a model of
JavaScript, demonstrate that it obeys two principles of
object capability systems
and show  how untrusted applications can be prevented from interfering with
the rest of the system. 
Recent programming languages 
\cite{CapJavaHayesAPLAS17,CapNetSocc17Eide,DOCaT14} including Newspeak
\cite{newspeak17}, Dart \cite{dart15}, Grace \cite{grace,graceClasses}
and Wyvern \cite{wyverncapabilities} have adopted the object
capability model.

\paragraph{Verification of Object Capability Programs}
Murray made the first attempt to formalise defensive consistency and
correctness~\cite{Murray10dphil}.  Murray's model was rooted in
counterfactual causation~\cite{Lewis_73}: an object is defensively
consistent when the addition of untrustworthy clients cannot cause
well-behaved clients to be given incorrect service.  Murray formalised
defensive consistency very abstractly, over models of (concurrent)
object-capability systems in the process algebra CSP~\cite{Hoare:CSP},
without a specification language for describing effects, such as what
it means for an object to provide incorrect service.  Both Miller and
Murray's definitions are intensional, describing what it means for an
object to be defensively consistent.

Dro\-sso\-pou\-lou and Noble \cite{capeFTfJP,capeFTfJP14} have
analysed Miller's Mint and Purse example \cite{MillerPhD} 
and discussed the six
capability policies 
as proposed in \cite{MillerPhD}.
In 
\cite{WAS-OOPSLA14-TR}, {they} 
sketched a  specification language,  \sd{used}  it to  
specify the six policies from \cite{MillerPhD}, 
\sd{showed} that several possible interpretations were possible, 
\sd{and} uncovered
the need for another four further policies.
They also
  sketched how 
a trust-sensitive 
example (the escrow exchange) could be verified in an open world
\cite{swapsies}. 
\sd{Their work does not support the concepts of control, time, or space, as in \Chainmail,
but it offers a primitive expressing trust.}
 
Swasey et al.\ \cite{dd}  have deployed
   \sd{powerful} 
  theoretical techniques to address similar problems:  
  \sd{They} show how step-indexing, Kripke worlds, and representing objects
as state machines with public and private transitions can be used to
reason about 
\sd{object capabilities}.
Devriese have demonstrated solutions to a range of exemplar problems,
including the DOM wrapper (replicated in our
section~\ref{sect:example:DOM}) and a mashup application.
\sd{Their} distinction
between public and private transitions 
\sd{is similar} to the
distinction between internal and external objects.

More recently, Swasey et al.\ \cite{ddd}  designed OCPL, a logic
for object capability patterns, that supports specifications and
proofs for object-oriented systems in an open world.  
\sd{They} 
draw on verification techniques for security and
information flow: separating internal implementations (``high values''
which must not be exposed to attacking code) from interface objects
(``low values'' which may be exposed).  OCPL supports defensive
consistency 
(\sd{they} use the term ``robust safety'' from the
security community \cite{Bengtson}) via a proof system that ensures
low values can never leak high values to external attackers. 
This means that low values \textit{can} be exposed to external code,
and the behaviour of the system is described by considering attacks only
on low values.  
\sd{They} use that logic to
prove a number of object-capability patterns, including
sealer/unsealer pairs, the caretaker, and a general membrane.

Schaefer et al.\ \cite{schaeferCbC} have recently
\sd{added}  support for information-flow security 
\sd{using} refinement to ensure correctness (in this case confidentiality) by
construction. 
By enforcing encapsulation, \sd{all} 
these approaches share similarity with techniques such as
ownership types \cite{ownalias,NobPotVitECOOP98}, which also
protect internal implementation objects from accesses that cross
encapsulation boundaries.  Banerjee and Naumann demonstrated that these
systems enforce representation independence (a property close to
``robust safety'') some time ago \cite{Banerjee:2005}.

\Chainmail\ differs from Swasey, Schaefer's, and Devriese's work in a number of ways:
\sd{They} are primarily concerned \sd{with} 
mechanisms that ensure encapsulation (aka 
confinement) while we abstract away from any mechanism via the
$\External{}$ predicate. 
\sd{They use powerful mathematical techniques
which  the users need  to understand in order to write their specifications,
while the \Chainmail users only need  to understand  first order logic and 
the holistic operators presented in this paper.}
\sd{ Finally, none of these systems offer the kinds of
holistic assertions addressing control flow, change, or temporal
operations that are at the core of \Chainmail's approach.
}

Scilla \cite{scillaOOPSLA19} is a minimalistic typed functional
language for writing smart contracts that compiles to the Ethereum
bytecode. Scilla's semantic model is restricted, assuming actor based
communication and restricting recursion,  thus facilitating static
analysis of Scilla contracts, and ensuring termination.
Scilla is able to demonstrate that a number of popular Ethereum
contracts avoid type errors, out-of-gas resource failures, and
preservation of virtual currency. 
Scilla's semantics are defined formally, but have not yet been represented in a
mechanised model.


Finally, the recent VerX tool is able to verify a range of
specifications for solidity contracts automatically \cite{VerX}.
Similar to \Chainmail, VerX has a specification language based on
temporal logic.  VerX offers three temporal operators (always, once,
prev) but only within a past modality, while \Chainmail\ has two
temporal operators, both existential, but with both past and future
modalities.   VerX specifications can also include predicates that
model the current invocation on a contract (similar to \Chainmail's
``calls''), can access variables, and compute sums (only) over
collections. \Chainmail\ is strictly more expressive as a
specification language, including quantification over objects and sets
(so can compute arbitrary reductions on collections) and of course
specifications for permission (``access''), space (``in'') and
viewpoint (``external'') which have no analogues in VerX. 
Unlike \Chainmail, VerX includes a practical tool that has
been used to verify   a hundred properties across case studies of
twelve Solidity contracts.

%% file: conclusion.tex
In this paper we have motivated the need for holistic specifications,
presented the specification language \Chainmail for writing such
specifications, and outlined the formal foundations of the language.
%
To focus on the key attributes of a holistic specification language,
\sd{we have kept  \Chainmail simple, only requiring an understanding of first order logic.}
\sd{We believe that the holistic features (permission, control, time, space and viewpoint),
are intuitive concepts 
when reasoning informally, and were pleased to have been able to provide their
formal semantics in what  we  argue is a simple manner.}

%% file: langOO.tex
\subsection{Modules and Classes}
\label{secONE}

\LangOO programs consist of modules, which are repositories of code. Since we study class based oo languages,
in this work, code is represented as classes, and  modules are mappings from  identifiers to class  descriptions.

\begin{definition}[Modules]
\label{defONE}
We define $\syntax{Module}$ as  the set of mappings from identifiers to class descriptions (the latter defined in Definition \ref{def:syntax:classes}):\\  

\begin{tabular}  {@{}l@{\,}c@{\,}ll}
\syntax{Module} \ \  &  \   $\triangleq $  \ &
   $ \{ \ \ \M \ \ \mid \ \  \M: \ \prg{Identifier} \   \longrightarrow \
  \  \syntax{ClassDescr}     \  \    \}$
 \end{tabular}
\end{definition}
 
Classes, as defined   below,
consist of field, method definitions and ghost field declarations.
 \LangOO is untyped, and therefore fields are declared without types, 
 method signatures and ghost field signatures consist of  sequences of parameters without types, and no return type.
 Method bodies consist of sequences of statements;
these can be field read or field assignments, object creation, method calls, and return statements.
All else, \eg booleans, conditionals, loops,  can be encoded.
Field read or write is only allowed \sd{if the object whose field is being read 
belongs to the same class as the current method. This is enforced by the operational semantics, \cf
Fig.  \ref{fig:Execution}.}
\sd{Ghost fields  are defined as implicit, side-effect-free functions with zero or more parameters. They are ghost information, \ie 
they are not directly stored in the objects, and are not read/written during execution. When such a ghostfield is
mentioned in an assertion, the corresponding function is evaluated. More in section \ref{sect:expressions}.
Note that the expressions that make up the bodies of ghostfield declarations (\prg{e}) are more complex than the terms that 
appear in individual statements.}

From now on we expect that the set of field and the set of ghostfields defined in a class are disjoint.  

\label{sec:syntax:classes}

\begin{definition}[Classes]
\label{def:syntax:classes}
Class descriptions consist of field declarations, method declarations, and ghostfield  declarations.
 
\begin{tabular}{lcll}
 \syntax{ClassDescr}   &   \BBC  &     \kwN{class}  \syntax{ClassId}    \lb\,  $($\ \ \syntax{FieldDecl} $)^*$ \
 $($  \syntax{MethDecl}\ $)^*$   \   $($   \syntax{GhosDecl}\ $)^*$ \ \ \rb
\\
\syntax{FieldDecl} &\BBC& \kwN{field} \f \\
\syntax{MethDecl} &\BBC&
     \kwN{method}\    \m\lp \x$^*$\rp     \lb\, \syntax{Stmts}  \,
    \rb
 \\
 \syntax{Stmts}  &\BBC&  \syntax{Stmt}     ~\SOR~  \syntax{Stmt} \semi \syntax{Stmts} \\
\syntax{Stmt}    &\BBC&
      \x.\f {\kw{:=}} \x   ~\SOR~  \x{\kw{:=}}  \x.\f    ~\SOR~        \x  {\kw{:=}} \x.\m\lp \x$^*$\rp     ~\SOR~     \x  {\kw{:=}}     \newKW\, \c\,\lp \x$^*$\rp   ~\SOR~
   \returnKW \,  \x   \\
  \syntax{GhostDecl} &\BBC&  \kwN{ghost} \f\lp \ \x$^*$\ \rp \lb \  \SE\ \rb\\
 \SE  &\BBC&    \kwN{true}   ~\SOR~  \kwN{false}   ~\SOR~  \kwN{null}  ~\SOR~  \x  \   ~\SOR~  
     \   \SE=\SE    ~\SOR~ \kwN{if}\, \SE\,   \kwN{then}\,  \SE\,    \kwN{else}\, \SE    ~\SOR~  \SE.\f\lp\ \SE$^*$ \ \rp\\
 \x, \f, \m &\BBC&  \prg{Identifier} 
 \end{tabular}

  \vspace{.03in}
  \noindent
 where we use metavariables as follows:
 $\x \in  \syntax{VarId} \ \ \  \f \in  \syntax{FldId} \ \ \  \m \in  \syntax{MethId} \ \ \  \c \in  \syntax{ClassId}$, and  \x\ includes \this
\end{definition}

We define a method lookup function, $\mathcal{M}$ which returns the corresponding method definition given a class \c\ and a method identifier \m, and similarly a ghostfield lookup function, $\mathcal{G}$ 

 \begin{definition}[Lookup] For a class identifier \prg{C}  and a method identifier \prg{m}$:$  $ ~ $ \\
\label{def:lookup}
\noindent
$
\Meths {} {\prg{C}} {m}       \triangleq  \ \left\{
\begin{array}{l}
                        \m\, \lp \p_1, ... \p_n \rp \lb Stmts   \rb\\
\hspace{0.5in} \mbox{if}\  \M(\prg{C}) =   \kwN{class}\, \prg{C}\, \  \lb ...   \kwN{method}\, \m\, \lp \p_1, ... \p_n \rp \lb Stmts  \rb  ... \ \rb.
\\
\mbox{undefined},  \ \ \ \mbox{otherwise.}
\end{array}
                    \right.$
\\
$
{\mathcal G} (\M, {\prg{C}}, {\f})    \ \   \triangleq  \ \left\{
\begin{array}{l}
                        \f\, \lp \p_1, ... \p_n \rp \lb \prg{e}  \rb\\
\hspace{0.5in} \mbox{if}\  \M(\prg{C}) =   \kwN{class}\, \prg{C}\, \  \lb ...   \kwN{ghost}\,  \m\, \lp \p_1, ... \p_n \rp \lb \prg{e}  \rb  ... \ \rb.
\\
\mbox{undefined},  \ \ \ \mbox{otherwise.}
\end{array}
                    \right.$
  \end{definition}

\subsection{The Operational Semantics of \LangOO}
\label{formal:semantics}

We will now define execution of \LangOO code.
We start by  defining the  runtime entities, and runtime configurations, $\sigma$, which consist of heaps and stacks of frames.
 The frames are pairs consisting of a continuation, and a mapping from identifiers to values.
The continuation represents the code to be executed next, and the mapping gives meaning
to the formal and local parameters.

\begin{definition}[Runtime Entities]
We define  addresses, values, frames, stacks, heaps and runtime configurations.

\begin{itemize}
\item
We take addresses to be an  enumerable set,  \prg{Addr}, and use the identifier $\alpha\in \prg{Addr}$ to indicate an address.
\item
Values, $v$, are either addresses, or sets of addresses or null:\\
 $~ ~ ~ \ v \in \{ \prg{null} \} \cup \prg{Addr}\cup {\mathcal P}( \prg{Addr})$.
\item
  Continuations are either   statements  (as defined in Definition~\ref{def:syntax:classes}) or a marker, \x {\kw{:=}} $\bullet$, for a nested call followed by
  statements to be executed
  once the call returns.

\begin{tabular}{lcll}
\syntax{Continuation} &\BBC&   \syntax{Stmts} ~\SOR~   \x {\kw{:=}} $\bullet$ \semi\ \syntax{Stmts} \\
 \end{tabular}

\item
Frames, $\phi$, consist of a code stub  and a  mapping from identifiers to values:\\  $~ ~ ~ \ \phi \ \in\ \syntax{CodeStub} \times \prg{Ident} \rightarrow Value$,
\item
Stacks,  $\psi$, are sequences of frames, $\psi\ ::=   \phi \ | \ \phi\cdot\psi$.
\item
Objects consist of a class identifier, and a partial mapping from field identifier to values: \\  \ $~ ~ ~ \ Object\ = \ \prg{ClassID} \times (\prg{FieldId} \rightarrow Value)$.
\item
Heaps, $\chi$, are mappings from addresses to objects:\  \  $\chi\ \in\ \prg{Addr} \rightarrow Object$.
\item
Runtime configurations, $\sigma$, are pairs of stacks and heaps, $\sigma\ ::=\ (\ \psi, \chi\ )$.
\end{itemize}

\end{definition}

Note that values may be sets of addresses. Such values are never part of the execution of 
\LangOO, but are used to give semantics to assertions . 
Next, we define the interpretation of variables (\x) and   field look up  (\x.\f)
in the context of frames,
heaps and runtime configurations; these interpretations are used to define the operational semantics and  also  the
validity of assertions, later on in Definitions 3-7. 

\begin{definition}[Interpretations]
\label{def:interp}
We first define lookup of fields and classes, where $\alpha$ is an address, and \f\, is a field identifier:
\begin{itemize}
\item
$\chi ({\alpha},{\f})$ $\triangleq$  $\fldMap({\alpha},{\f})$\ \ \ if \ \ $\chi(\alpha)=(\_, \fldMap)$.
\item
$\ClassOf {\alpha} {\chi} $ $\triangleq$ $\c$\  \ \ if \ \ $\chi(\alpha)=(\c,\_)$
\end{itemize}

\noindent
We now define interpretations  as follows:

\begin{itemize}
\item
$\interp {\x}{\phi} $ $\triangleq$ $\phi(\x)$
\item
$\interp {\x.\f}{(\phi,\chi)} $ $\triangleq$ $v$, \ \ \ if \ \ $\chi(\phi(\x))=(\_, \fldMap)$ and $\fldMap(\f)$=$v$

\end{itemize}

\noindent
For ease of notation, we also use the shorthands below:
\begin{itemize}
\item
$\interp {\x}{(\phi\cdot\psi,\chi)} $ $\triangleq$ $\interp {\x}{\phi} $
\item
$\interp {\x.\f}{(\phi\cdot\psi,\chi)} $ $\triangleq$ $\interp  {\x.\f}{(\phi,\chi)} $
\item
$\ClassOf {\alpha} {(\psi,\chi)} $ $\triangleq$ $\ClassOf {\alpha} {\chi} $
\item
$\ClassOf {\x} {\sigma} $ $\triangleq$ $\ClassOf {\interp {\x}{\sigma}} {\sigma} $
\end{itemize}

\end{definition}

In the definition of the operational semantics of \LangOO we use the following notations for lookup and updates of runtime entities :

\begin{definition}[Lookup and update of runtime configurations]
We define convenient shorthands for looking up in  runtime entities.
\begin{itemize}
\item
Assuming that $\phi$ is the tuple  $(\prg{stub}, varMap)$, we use the notation  $\phi.\prg{contn}$ to obtain \prg{stub}.
\item
Assuming a value v, and that $\phi$ is the tuple  $(\prg{stub}, varMap)$, we define $\phi[\prg{contn}\mapsto\prg{stub'}]$ for updating the stub, \ie
$(\prg{stub'}, varMap)$.   We use  $\phi[\x \mapsto v]$  for updating the variable map, \ie  $(\prg{stub}, varMap[\x \mapsto v])$.
\item
Assuming a heap $\chi$, a value $v$, and   that $\chi(\alpha)=(\c, fieldMap)$,
we use $\chi[\alpha,\f \mapsto v]$ as a shorthand for updating the object, \ie $\chi[\alpha \mapsto (\c, fieldMap[\f \mapsto v]]$.
\end{itemize}

\end{definition}

\begin{figure*}
$\begin{array}{l}
\inferenceruleNN {methCall\_OS} {
\\
\phi.\prg{contn}\ =\ \x {\kw{:= }} \x_0.\m \lp \x_1, ... \x_n \rp \semi \prg{Stmts}
\\
\interp{\x_0}{\phi} = \alpha
\\
\Meths {} {\ClassOf {\alpha} {\chi}} {\m} \  =  \ \m\lp \p_1, \ldots \p_n \rp \lb \prg{Stmts}_1   \rb
  \\
 \phi''\ =\  (\  \prg{Stmts}_1,\ \ (\ \this \mapsto \alpha,
  \p_1 \mapsto  \interp{\x_1}{\phi}, \ldots \p_n \mapsto  \interp{\x_n}{\phi}\ ) \ )
}
{
 \M,\, (\ \phi\cdot\psi,\ \chi\ )\ \ \leadsto\  \ (\ \phi''\cdot\phi[\prg{contn}\mapsto\x  \kw{:=} \bullet \semi \prg{Stmts}] \cdot\psi,\ \chi\ )
}

\\ \\
\inferenceruleNN {varAssgn\_OS} {
 \phi.\prg{contn} \ = \ \x  {\kw{:= }}  \y.\f \ \semi \prg{Stmts}\ \hspace{2cm} \ClassOf {\y} {\sigma} =\ClassOf {\this} {\sigma}
}
{
 \M,\,  (\ \phi\cdot\psi, \chi\ )\ \ \leadsto\  \ (\ \phi[ \prg{contn} \mapsto \prg{Stmts}, \x\mapsto \interp{\y.\f}{\phi,\chi}] \cdot\psi,\ \chi\  )
}
\\
\\
\inferenceruleNN{fieldAssgn\_OS} {
 \phi.\prg{contn}\ =\  \x.\f  \kw{:=} \y  \semi \prg{Stmts} \hspace{2cm} \ClassOf {\x} {\sigma} =\ClassOf {\this} {\sigma}
}
{
 \M,\,  (\ \phi\cdot\psi, \chi\  )\ \ \leadsto\  \ (\ \phi[\prg{contn}\mapsto  \prg{Stmts} ] \cdot\psi, \chi[\interp{\x}{\phi},\f \mapsto \interp{\y}{\phi,\chi}]\  )
}
\\
\\
\inferenceruleNN {objCreate\_OS} {
 \phi.\prg{contn}\ =\  \x  \kw{:=} \kwN{new }\, \c \lp \x_1, ... \x_n \rp  \semi \prg{Stmts}
 \\
 \alpha\ \mbox{new in}\ \chi
 \\
\f_1, .. \f_n\ \mbox{are the fields declared in } \M(\c)
}
{
 \M,\,  (\ \phi\cdot\psi, \chi\ )\ \ \leadsto\  \ (\ \phi[\prg{contn}\mapsto  \prg{Stmts},\x \mapsto \alpha\ ] \cdot\psi, \ \chi[\alpha \mapsto (\c,\f_1 \mapsto \interp{\x_1}{\phi},  ... \f_n \mapsto \interp{\x_n}{\phi}  ) ]\ )
}
\\
\\
\inferenceruleNN {return\_OS} {
 \phi.\prg{contn}\ =\   {\kwN{return }}\, \x  \semi \prg{Stmts}\ \  \ or\  \ \  \phi.\prg{contn}\ =\   {\kwN{return}}\, \x
 \\
\phi'.\prg{contn}\ =\  \x' \kw{:=} \bullet  \semi \prg{Stmts}'
}
{
 \M,\,  (\ \phi\cdot\phi'\cdot\psi, \chi\ )\ \ \leadsto\  \ (\ \phi'[\prg{contn}\mapsto  \prg{Stmts'},\x' \mapsto \interp{\x}{\phi}] \cdot\psi, \ \chi \ )
}
\end{array}
$
\caption{Operational Semantics}
\label{fig:Execution}
\end{figure*}

Execution of a statement has the form $\M, \sigma \leadsto \sigma'$, and is defined in figure \ref{fig:Execution}.

\begin{definition}[Execution] of one or more steps is defined as follows:

\begin{itemize}
     \item
   The relation $\M, \sigma \leadsto \sigma'$, it is defined in Figure~\ref{fig:Execution}.

   \item
   $\M, \sigma \leadsto^* \sigma'$ holds, if i) $\sigma$=$\sigma'$, or ii) there exists a $\sigma''$ such that
   $\M, \sigma \leadsto^* \sigma''$ and $\M, \sigma'' \leadsto \sigma'$.
 \end{itemize}

\end{definition}

\subsection{Module linking}

When studying validity of assertions in the open world we are concerned with whether   the  module
under consideration makes  a certain guarantee when executed in conjunction with other modules. To answer this, we
 need the concept of linking other modules to the module  under consideration.
 Linking, $\link$ ,  is an operation that takes two modules, and creates a module which corresponds  to the union of the two.
We place some conditions for module linking to be defined: We require that the two modules do not contain implementations for the same class identifiers,


\begin{definition}[Module Linking]
\label{def:link}
The linking operator\  \ $\link:\  \syntax{Module} \times  \syntax{Module} \longrightarrow \syntax{Module}$ is defined as follows:

$
\M \link \M{'}  \ \triangleq  \ \ \left\{
\begin{array}{l}
                        \M\ \link\!_{aux}\ \M{'},\ \ \   \hbox{if}\  \ dom(\M)\!\cap\!dom(\M')\!=\!\emptyset\\
\mbox{undefined}  \ \ \ \mbox{otherwise.}
\end{array}
                    \right.$

and where,
\begin{itemize}
     \item
   For all  $\prg{C}$: \ \
   $(\M\ \link\!_{aux}\ \M')(\prg{C})\  \triangleq  \ \M(\prg{C})$  if  $\prg{C}\in dom(\M)$, and  $\M'(\prg{C})$ otherwise.
 \end{itemize}
\end{definition}

Some properties of linking are described in lemma \ref{lemma:linking} in the main text. \sophia{
%
%
%
%
 For the proof, 
\sd{ (1) and (2) follow from Definition \ref{def:link}. (3) follows from \ref{def:link}, and the fact that if a lookup $\mathcal \M$ is
defined for $\M$, then it is also defined for $\M\link\M'$ and returns the same method, and similar result for class lookup.}
}

 \subsection{Module pairs and visible states semantics}

A module $\M$ adheres to an invariant assertion  $\A$, if it satisfies
$\A$ in all runtime configurations that  can be reached through execution of the code of $\M$ when linked to that
of {\em any other} module $\M'$, and
which are {\em external} to $\M$. We call external to $\M$ those
configurations which are currently executing code which does not come from $\M$. This allows the code in $\M$ to break
the invariant internally and temporarily, provided that the invariant is observed across the states visible to the external client $\M'$.

We have defined two module execution in the main paper, Def. \ref{def:execution:internal:external}.
%
%
%
%

In \sophia{that} definition 
$n$ is allowed to have the value $2$. In this case the final bullet is trivial and  there exists a direct, external transition from $\sigma$ to $\sigma'$.  Our definition is related to the concept of visible states semantics, but differs in that visible states semantics select the configurations at which an invariant is expected to hold, while we select the states which are considered for executions which are expected to satisfy an invariant. Our assertions can talk about several states (through the use of the $\Future {\_}$ and $\Past{\_}$ connectives), and thus, the intention of ignoring some intermediate configurations can only be achieved if we refine the concept of execution. 

%
%
%
%

\sophia{We have defined initial and arising configurations  in Definition \ref{def:arise}. Note that there are infinitely many different initial configurations, they will be differing in the code stored in the continuation of the unique frame.}

%
%
%
%

%% file: assertionsAppendix.tex
We now define the syntax and semantics of expressions and holistic assertions.
\sd{The novel, holistic, features of \Chainmail (permission, control, time, space, and viewpoint),
as well as our wish to support some form of recursion while keeping the logic of assertions classical,  introduced 
challenges, which we discuss in this section.}

 \subsection{Syntax of Assertions}

\begin{definition}[Assertions]  \sd{Assertions consist of (pure) expressions \e, classical assertions about the contents of heap/stack, the usual logical  connectives, as well as our holistic concepts.}
\label{def:assertions}

 $\begin{array}{lcl}
  ~  \\
 \SE  &\BBC&    \kwN{true}   ~\SOR~  \kwN{false}   ~\SOR~  \kwN{null}  ~\SOR~  \x  \   ~\SOR~  
     \   \SE=\SE    ~\SOR~ \kwN{if}\, \SE\,   \kwN{then}\,  \SE\,    \kwN{else}\, \SE    ~\SOR~  \SE.\f\lp\ \SE^* \ \rp\\
     \\
 \A &\ \BBC   &   \SE \   \mid \  \SE=\SE  \mid \   \SE:\prg{ClassId}  \ \mid \
    \SE\in\prg{S}   \mid  \  \\
    &
  &  \A \rightarrow \A  \ \mid\  \     \A \wedge \A  \ \mid\  \  \A \vee \A  \ \mid\  \ \neg A   \ \mid \ \\
  & &  \forall \x.\A  \ \mid \  \forall \prg{S}:SET.\A  \ \mid  \  \exists \x.\A  \ \mid \  \exists \prg{S}:SET.\A  \  \ \mid\   \\
 &    &  \CanAccess x y 
           \ \mid\  \Calls {\prg{x}}  {\prg{m}} {\prg{x}}  {\prg{x}^*}\\          
&    &  \Next \A  \ \mid \   \Future \A \ \mid \  \Prev \A   \ \mid \ \Past \A \ \mid \\  
 &    &        \Using \SF  \A  \ \mid \  \External \x     \\
 \\
 \x, \f, \m &\BBC&  \prg{Identifier}  ~ \\
\end{array}$
\end{definition}
 
 \sd{Expressions support calls with parameters  ($\e.\f(\e^*)$); these are calls to ghostfield
functions. This  supports recursion at the level of expressions; therefore, the value of  an expression  may be
undefined (either because of infinite recursion, or because the expression accessed undefined fields or variables). 
Assertions of the form   $\e$=$\e'$ are satisfied only if both $\e$ and $\e'$ are defined. Because we do not support 
recursion at the level of assertions, assertions from a classical logic (\eg $\A \vee \neg\A$ is a tautology). }
 
We will discuss evaluation of expressions in section \ref{sect:expressions}, standard assertions about heap/stack and logical
 connectives in \ref{sect:standard}. 
 \sophia{We have discussed  the treatment of  permission, control, space, and viewpoint  in 
the main text in the  Definitions 3-7  in section \ref{sect:pcsv} 
the treatment of time in  Definitions 8,9 in the main text, section \ref{sect:time},
We will discuss properties of assertions in Lemmas \ref{lemma:classic}-\ref{lemma:classic:two}.}
 \sd{The judgement $\M\mkpair \M', \sigma  \models \A$ expresses that $\A$ holds in  $\M\mkpair \M'$ and $\sigma$, and 
while $\M\mkpair \M', \sigma  \not\models \A$  expresses that $\A$ does not hold  in  $\M\mkpair \M'$ and $\sigma$.}

\subsection{Values of Expressions}
\label{sect:expressions}

The value  of  an expression  is described through judgment $ \M,\, \sigma, \SE \ \hookrightarrow\  v$,
defined in  Figure \ref{fig:ValueSimpleExpressions}.
We use the configuration, $\sigma$, to read the contents of the top stack frame
(rule ${\sf {Var\_Val}}$) or the contents of the heap (rule
${\sf {Field\_Heap\_Val}}$). We use the module, \M, to find the  ghost field declaration corresponding to the
ghost field being used.

The treatment of fields and ghost fields is described in rules ${\sf {Field\_Heap\_Val}}$,\\  ${\sf {Field\_Ghost\_Val}}$ and 
${\sf {Field\_Ghost\_Val2}}$.  If the field \f~ exists in the heap, then its value is returned (${\sf {Field\_Heap\_Val}}$). 
Ghost field reads, on the other hand, have the form $\e_0.\f(\e_1,...\e_n)$, and their value is
described in rule ${\sf {Field\_Ghost\_Val}}$:
The lookup function $\mathcal G$  (defined in the obvious way in the Appendix, Def.\ref{def:lookup})
returns the expression constituting the body for that ghost field, as defined in the class of $\e_0$.
We return  that expression
evaluated in a configuration where the formal parameters have been substituted by the values of the actual
parameters.

Ghost fields support recursive definitions. For example, imagine a module $\M_0$ with
a class \prg{Node} which has a field called \prg{next}, and which 
had a ghost field \prg{last}, which finds  the last \prg{Node} in a sequence
and is defined recursively as \\
$~ \strut \hspace{.1cm}$ \ \ \ \prg{if}\ \ \prg{this.next}=\prg{null}\  \prg{then} \ \prg{this} \ \prg{else} \ \prg{this.next.last},\\
and another ghost field \prg{acyclic}, which expresses that a sequence is acyclic,
defined recursively as \\
$~ \strut \hspace{.1cm}$ \ \ \ \prg{if}\ \ \prg{this.next}=\prg{null}\  \prg{then} \ \prg{true} \ \prg{else} \ \prg{this.next.acyclic}.\\

The relation $ \hookrightarrow$ is partial. 
For example, assume   a configuration
$\sigma_0$ where
\prg{acyc} points to a \prg{Node} whose field \prg{next} has value \prg{null}, and   
\prg{cyc} points to a \prg{Node} whose field \prg{next} has the same value as \prg{cyc}. Then,   
$\M_0,\sigma_0,\,\prg{acyc.acyclic}  \ \hookrightarrow\  \prg{true}$, but we would have no value for 
$\M_0,\sigma_0,\, \prg{cyc.last}  \ \hookrightarrow\  ...$, nor for
$\M_0,\sigma_0,\, \prg{cyc.acyclic}  \ \hookrightarrow\  ...$.

Notice also that for an expression of the form  
\prg{\e.\f}, both ${\sf {Field\_Heap\_Val}}$ and ${\sf {Field\_Ghost\_Val2}}$ could be applicable: rule ${\sf {Field\_Heap\_Val}}$
will be applied if \prg{f} is a field of the object at \prg{e}, while rule ${\sf {Field\_Ghost\_Val}}$
will be applied if \prg{f} is a ghost field of the object at \prg{e}. We expect the set of fields and ghost fields in a 
given class to be disjoint.
This allows a specification to be agnostic over whether a field is a physical field or just ghost information.
For example, assertions (1) and (2) from  section  \ref{sect:motivate:Bank}
 talk about the \prg{balance} of an \prg{Account}. 
In module $\M_{BA1}$ (Appendix~\ref{Bank:appendix}), where we keep the balances in the account objects, this is a physical field. 
In $\M_{BA2}$ (also in Appendix~\ref{Bank:appendix}), where we keep the
balances in a ledger, this is ghost information.

\begin{figure*}
{$\begin{array}{l}
\begin{array}{llll}
\inferenceruleN {True\_Val} {
}
{
 \M,\, \sigma, \kwN{true} \ \hookrightarrow\  \kwN{true}
}
& 
\inferenceruleN {False\_Val} {
}
{
 \M,\, \sigma, \kwN{false} \ \hookrightarrow\  \kwN{false}
}
&
\inferenceruleN  {Null\_Val} {
}
{
 \M,\, \sigma, \kwN{null} \ \hookrightarrow\  \kwN{null}
}
&
\inferenceruleN {Var\_Val} {
}
{
 \M,\,  \sigma, \x \ \hookrightarrow\   \sigma({\x})
}
\end{array}
\\ \\
\begin{array}{lll}
\begin{array}{l}
\inferenceruleNM{Field\_Heap\_Val} {
 \M,\,  \sigma, \SE \ \hookrightarrow\   \alpha \hspace{1.5cm} 
 \sigma(\alpha,\f)=v
}
{
 \M,\, \sigma, \SE.\f  \ \hookrightarrow\   v
}
\\
\\
\inferenceruleNM{Field\_Ghost\_Val2} {
 \M,\, \sigma, \SE.\f \lp \rp \ \hookrightarrow\   v
}
{
 \M,\, \sigma, \SE.\f   \ \hookrightarrow\   v
}
\end{array}
& &
\inferenceruleNM{Field\_Ghost\_Val}
{
~ \\
 \M,\, \sigma, \SE_0   \ \hookrightarrow\  \alpha
\\
 \M,\, \sigma, \SE_i  \ \hookrightarrow\   v_i\ \ \ \ i\in\{1..n\}
 \\
{\mathcal{G}}
(\M, {\ClassOf {\alpha} {\sigma}}, {\f}) \  =  
\ \f\lp \p_1, \ldots \p_n \rp \lb\ \SE \ \rb
  \\
  \M,\,\sigma[\p_1\mapsto v_1, .... \p_n\mapsto v_n], \SE    \hookrightarrow_{\SAF}\   v 
 }
{
 \M,\,  \sigma, \ \SE_0.\f \lp \SE_1,....\SE_n\rp \hookrightarrow   \ v
}
\\ \\
\inferenceruleNM{If\_True\_Val} 
{
 \M,\,  \sigma, \SE \ \hookrightarrow\   \prg{true}  \\
   \M,\,  \sigma, \SE_1 \ \hookrightarrow\   v  
}
{
 \M,\, \sigma, \kwN{if}\ \SE\  \kwN{then} \ \SE_1 \ \kwN{else} \ \SE_2\  \hookrightarrow  \ v
}
& &
\inferenceruleNM {If\_False\_Val} 
{
 \M,\,  \sigma, \SE \ \hookrightarrow\   \prg{false}  \\
   \M,\,  \sigma, \SE_2 \ \hookrightarrow\   v  }
{
 \M,\, \sigma, \kwN{if}\ \SE\  \kwN{then} \ \SE_1 \ \kwN{else} \ \SE_2\  \hookrightarrow\  v
}
\\ \\ 
\inferenceruleNM {Equals\_True\_Val} 
{
 \M,\,  \sigma, \SE_1 \ \hookrightarrow\    v \\
   \M,\,  \sigma, \SE_2 \ \hookrightarrow\     v 
}
{
 \M,\, \sigma, \SE_1 =  \SE_2 \hookrightarrow \prg{true}
}
& &
\inferenceruleNM {Equals\_False\_Val} 
{
 \M,\,  \sigma, \SE_1 \ \hookrightarrow\    v \\
   \M,\,  \sigma, \SE_2 \ \hookrightarrow\     v' \hspace{2cm}  v\neq v'
}
{
 \M,\, \sigma, \SE_1 =  \SE_2 \hookrightarrow \ \prg{false}
}
\end{array}
\end{array}
$}
\caption{Value of  Expressions}
\label{fig:ValueSimpleExpressions}
\end{figure*}

\subsection{Satisfaction of Assertions - standard}
\label{sect:standard}
\sd{
We now define the semantics of assertions involving expressions, the heap/stack, and logical connectives.
The semantics are unsurprising, except, perhaps the relation between validity of assertions and the values of
expressions.
}

 \begin{definition}[Interpretations for simple expressions]

For a runtime configuration, $\sigma$,    variables $\x$ or \SF, we define its interpretation as follows:

\begin{itemize}
  \item
  $\interp {\x}{\sigma}$ $ \triangleq$ $\phi(\x)$  \ \ if \ \ $\sigma$=$(\phi\cdot\_,\_)$
  \item
  $\interp {\SF}{\sigma}$ $ \triangleq$ $\phi(\SF)$  \ \ if \ \ $\sigma$=$(\phi\cdot\_,\_)$
  \item
    $\interp {\x.\f}{\sigma}$ $ \triangleq$ $\chi(\interp {\x}{\sigma},\f)$  \ \ if \ \ $\sigma$=$(\_,\chi)$
   \end{itemize}
\end{definition}

\begin{definition}[ Basic Assertions] For modules $\M$, $\M'$,  configuration $\sigma$,  we define$:$
\label{def:valid:assertion:basic}
\begin{itemize}
\item
$\M\mkpair \M', \sigma \models\SE$ \IFF   $ \M,\,  \sigma, \SE \ \hookrightarrow\   \prg{true}$ 
\item
$\M\mkpair \M', \sigma \models\SE=\SEPrime$ \IFF there exists a value $v$ such that  $\M,\,  \sigma, \SE \ \hookrightarrow\   v$  and $ \M,\,  \sigma, \SEPrime \ \hookrightarrow\   v$.
           \item
$\M\mkpair \M', \sigma \models\SE:\prg{ClassId}$ \IFF there exists an address $\alpha$ such that \\
$\strut ~ $ \hspace{2in} \hfill   
 $ \M,\,  \sigma, \SE \ \hookrightarrow\   \alpha$, and $\ClassOf{\alpha}{\sigma}$ = \prg{ClassId}.
\item
$\M\mkpair \M', \sigma \models \SE\in \prg{S}$ \IFF there exists a value $v$ such that 
 $ \M,\,  \sigma, \SE \ \hookrightarrow\   v$, and $v \in \interp{\prg{S}}{\sigma}$.
\end{itemize}
\end{definition}

Satisfaction of assertions which contain expressions is predicated on termination of these expressions.
Continuing our earlier example,  
$\M_0\mkpair \M', \sigma_0 \models \prg{acyc.acyclic}$ holds for any $\M'$, while $\M_0\mkpair \M', \sigma_0 \models \prg{cyc.acyclic}$
does not hold, and $\M_0\mkpair \M', \sigma_0 \models \prg{cyc.acyclic}=\prg{false}$ does not hold either.
In general, when $\M\mkpair \M', \sigma  \models \prg{e}$ holds,  then $\M\mkpair \M', \sigma  \models \prg{e}=\prg{true}$ holds too.
But when $\M\mkpair \M', \sigma  \models \prg{e}$ does not hold, this does \emph{not} imply that $\M\mkpair \M', \sigma  \models \prg{e}=\prg{false}$ holds.
Finally, an assertion of the form $\e_0=\e_0$ does not always hold; for example,   $\M_0\mkpair \M', \sigma_0 \models \prg{cyc.last}=\prg{cyc.last}$ does not hold.

We now define satisfaction of assertions which involve logical connectives and existential or universal quantifiers, in the standard way:

\begin{definition}[Assertions with logical connectives and quantifiers]  
\label{def:valid:assertion:logical}
For modules $\M$, $\M'$, assertions $\A$, $\A'$, variables \prg{x}, \prg{y}, \prg{S},  and configuration $\sigma$, we define$:$
\begin{itemize}
\item
$\M\mkpair \M', \sigma \models \forall \prg{S}:\prg{SET}.\A$ \IFF  $\M\mkpair \M', \sigma[\prg{Q}\mapsto R] \models  \A[\prg{S}/\prg{Q}]$ \\
$\strut ~ $ \hfill for all sets of addresses $R\subseteq dom(\sigma)$, and  all \prg{Q} free in $\sigma$ and $\A$.
\item
$\M\mkpair \M', \sigma \models \exists \prg{S}:\prg{SET}\!.\,\A$ \IFF  $\M\mkpair \M', \sigma[\prg{Q}\mapsto R] \models  \A[\prg{S}/\prg{Q}]$ \\
 $\strut ~ $ \hfill  for some set of addresses $R\subseteq dom(\sigma)$, and   \prg{Q} free in $\sigma$ and $\A$.
\item
$\M\mkpair \M', \sigma \models \forall \prg{x}.\A$ \IFF
$\sigma[\prg{z}\mapsto \alpha] \models  \A[\prg{x}/\prg{z}]$ \ for all  $\alpha\in dom(\sigma)$, and  some \prg{z} free in $\sigma$ and $\A$.
\item
$\M\mkpair \M', \sigma \models \exists \prg{x}.\A$ \IFF
$\M\mkpair \M', \sigma[\prg{z}\mapsto \alpha] \models  \A[\prg{x}/\prg{z}]$\\
$\strut ~ $ \hfill for some  $\alpha\in dom(\sigma)$, and   \prg{z} free in $\sigma$ and $\A$.
\item
$\M\mkpair \M', \sigma \models \A \rightarrow \A' $ \IFF  $\M\mkpair \M', \sigma \models \A $ implies $\M\mkpair \M', \sigma \models \A' $
\item
$\M\mkpair \M', \sigma \models  \A \wedge \A'$   \IFF  $\M\mkpair \M', \sigma \models  \A $
and $\M\mkpair \M', \sigma \models  \A'$.
\item
$\M\mkpair \M', \sigma \models  \A \vee \A'$   \IFF  $\M\mkpair \M', \sigma \models  \A $
or $\M\mkpair \M', \sigma \models  \A'$.
\item
$\M\mkpair \M', \sigma \models  \neg\A$   \IFF  $\M\mkpair \M', \sigma \models  \A $
does not hold.
\end{itemize}
\end{definition}

Satisfaction is not preserved with growing configurations; for example, the assertion $\forall \x. [\ \x : \prg{Account} \rightarrow \x.\prg{balance}>100\ ]$ 
may hold in a smaller configuration, but not hold in an extended configuration. 
Nor is it preserved with configurations getting smaller; consider \eg $\exists \x. [\ \x : \prg{Account} \wedge \x.\prg{balance}>100\ ]$.

\noindent
Again, with our earlier example,  
$\M_0\mkpair \M', \sigma_0 \models \neg (\prg{cyc.acyclic}=\prg{true})$    and  
$\M_0\mkpair \M', \sigma_0 \models  \neg (\prg{cyc.acyclic}=\prg{false})$, 
and also 
$\M_0\mkpair \M', \sigma_0 \models  \neg (\prg{cyc.last}=\prg{cyc.last})$
hold.

\label{sect:pl} 
We define equivalence of  assertions in the usual sense: two assertions are equivalent if they are satisfied  in
the context of the same configurations.
Similarly, an assertion entails another assertion, iff all configurations 
which satisfy the former also satisfy the latter.  

\begin{definition}[Equivalence and entailments of assertions]
$ ~ $

\begin{itemize}
\item
$\A \subseteqq \A'\  \IFF\    \forall \sigma.\, \forall \M, \M'. \ [\ \ \M\mkpair \M', \sigma \models \A\ \mbox{ implies }\ \M\mkpair \M', \sigma \models \A'\ \ ].$
\item
$\A \equiv \A'\  \IFF\     \A \subseteqq \A' \mbox{ and }  \A' \subseteqq \A.$
\end{itemize}
\end{definition}

\begin{lemma}[Assertions are classical-1]
\label{lemma:classic}
For all runtime configurations $\sigma$,    assertions $\A$ and $\A'$, and modules $\M$  and $\M'$, we have
\begin{enumerate}
\item
$\M\mkpair \M', \sigma \models \A$\ or\ $\M\mkpair \M', \sigma \models \neg\A$
\item
$\M\mkpair \M', \sigma  \models \A \wedge \A'$ \SP if and only if \SP $\M\mkpair \M', \sigma \models \A$ and $\M\mkpair \M', \sigma  \models \A'$
\item
$\M\mkpair \M', \sigma  \models \A \vee \A'$ \SP if and only if \SP $\M\mkpair \M', \sigma  \models \A$ or  $\sigma \models \A'$
\item
$\M\mkpair \M', \sigma  \models \A \wedge \neg\A$ never holds.
\item
$\M\mkpair \M', \sigma  \models \A$ and  $\M\mkpair \M', \sigma  \models \A \rightarrow \A'$  implies
$\M\mkpair \M', \sigma  \models \A '$.
\end{enumerate}
\end{lemma}
\begin{proof}  The proof of part (1) requires to first prove that for all \emph{basic assertions} \A, \\
\strut \hspace{1.1cm} (*) \ \ \ either $\M\mkpair \M', \sigma  \models \A$
or $\M\mkpair \M', \sigma  \not\models \A$.\\
We prove this using Definition \ref{def:valid:assertion:basic}.
 Then, we prove (*) for all
possible assertions, by induction of the structure of \A, and the Definitions 
 \ref{def:valid:assertion:logical},
 and also Definitions
  \ref{def:valid:assertion:access}, \ref{def:valid:assertion:control}, \ref{def:valid:assertion:view},  
 \ref{def:valid:assertion:space}, and \ref{def:valid:assertion:time}.
Using the definition of $\M\mkpair \M', \sigma \models \neg\A$ from Definition  \ref{def:valid:assertion:logical} we conclude the proof of (1).

For parts  (2)-(5) the proof goes by application of the corresponding definitions from \ref{def:valid:assertion:logical}.
Compare also with appendix \ref{sect:coq}.
 
  \end{proof}.
 
 \begin{lemma}[Assertions are classical-2]
 \label{lemma:classic:two}
For     assertions $\A$, $\A'$, and $\A''$ the following equivalences hold
\label{lemma:basic_assertions_classical}
\begin{enumerate}
\item
$ \A \wedge\neg \A \ \equiv \  \prg{false}$
\item
$ \A \vee \neg\A   \ \equiv \  \prg{true}$
\item
$ \A \wedge \A'  \ \equiv \  \A' \wedge \A$
\item
$ \A \vee \A'  \ \equiv \  \A' \vee \A$
\item
$(\A \vee \A') \vee \A'' \ \equiv \  \A \vee (\A' \vee\A'')$
\item
$(\A \vee \A') \wedge \A'' \ \equiv \  (\A \wedge \A')\, \vee\, (\A \wedge \A'')$
\item
$(\A \wedge \A') \vee \A'' \ \equiv \  (\A \vee \A')\, \wedge\, (\A \vee \A'')$
\item
$\neg (\A \wedge \A') \  \ \equiv \  \neg  \A   \vee\, \neg \A''$
\item
$\neg (\A \vee \A') \  \ \equiv \  \neg  \A   \wedge\, \neg \A'$
\item
$\neg (\exists \prg{x}.\A )  \  \ \equiv \  \forall \prg{x}.(\neg  \A)$
\item
$\neg (\exists \prg{S}:\prg{SET}.\A )  \  \ \equiv \  \forall \prg{S}:\prg{SET}.(\neg  \A)$
\item
$\neg (\forall \prg{x}. \A)  \  \ \equiv \  \  \exists \prg{x}.\neg(\A )$
\item
$\neg (\forall \prg{S}:\prg{SET}. \A)  \  \ \equiv \  \  \exists \prg{S}:\prg{SET}.\neg(\A )$
\end{enumerate}
\end{lemma}
\begin{proof}
All points follow by application of the corresponding definitions from \ref{def:valid:assertion:logical}. 
Compare also with appendix \ref{sect:coq}.
 \end{proof}

%
%

%% file: Bank.appendix.tex
\begin{figure}[thb]
\begin{lstlisting}
class Bank{

   method newAccount(amt){
        if (amt>=0) then{
            return new Account(this,amt)
   }   }
}

class Account{

    field balance
    field myBank
    
    method deposit(src,amt){
       if (amt>=0 && src.myBank=this.myBank && src.balance>=amt) then{
           this.balance = this.balance+amt
           src.balance = src.balance-amt
   }   }
   method makeAccount(amt){
       if (amt>=0 && this.balance>=amt) then{
           this.balance = this.balance - amt;
           return new Account(this.myBank,amt)
   }    }
}
\end{lstlisting}
 \vspace*{-7mm}
\caption{$M_{BA1}$: Implementation of \prg{Bank} and \prg{Account} -- version 1}
\label{fig:BanAccImplV1}
\end{figure}

In this section we revisit the \prg{Bank}/\prg{Account} example from
  \ref{sect:motivate:Bank}, 
 and show two different
 implementations, derived from Noble et al. \cite{arnd18} . 
 Both implementations  satisfy the three functional specifications and the holistic assertions
 (1), (2) and (3)  shown in section  \ref{sect:motivate:Bank}.
 The first version gives rise to runtime configurations as $\sigma_1$, 
 shown on the left side of Fig.   \ref{fig:BankAccountDiagrams}, 
 while the
 second version gives rise to runtime configurations as $\sigma_2$,
 shown on the right side of Fig.   \ref{fig:BankAccountDiagrams}. 
 in the main text.

 In this code, we use more syntax than the minimal syntax defined for \LangOO in Def. \ref{defONE}, as we use conditionals, and we allow nesting of expressions, e.g.\ a field read to be the receiver of a method call. Such extension can easily be encoded in the base syntax.

$\M_{BA1}$, the fist version is shown Fig. \ref{fig:BanAccImplV1}. It keeps all the information in the \prg{Account} object: namely,
the \prg{Account} contains the pointer to the bank, and the balance, while the \prg{Bank} is a pure capability, which contains
no state but is necessary for the creation of new \prg{Account}s.
In this version we have no ghost fields.

\begin{figure}[htb]
\begin{lstlisting}
class Bank{
   field ledger // a Node
   
    method deposit(dest,src,amt){
       destNd = this.ledger.find(dest)
       srcNd = this.ledger.find(src)
       srcBalance = srcNd.getBalance()
       if ( destNd =/=null && srcNd=/=null && srcBalance>=amt && amt >=0 ) then
           destNd.addToBalance(amt)
           srcNc.addToBalance(-amt)           
    }  }     
    method newAccount(amt){
      if (amt>=0)  then{
           newAcc = new Account(this);
           this.ledger = new Node(amt,this.ledger,newAcc)
           return newAcc 
    } }
   
   ghost  balance(acc){ this.ledger.balance(acc)  } 
}
\end{lstlisting}
 \vspace*{-7mm}
\caption{$M_{BA2}$: Implementation of \prg{Bank}   -- version 2}
\label{fig:BanAccImplV2b}
\end{figure}

\begin{figure}[htb]
\begin{lstlisting}
class Node{
   field balance
   field next   
   field myAccount
   
   method addToBalance(amt){
       this.balance = this.balance + amt
   }   
   method find(acc){
      if this.myAccount == acc then{
          return this
     } else { 
          if this.next==null then{
              return null
          } else {
              return this.next.find(acc)
    }  } } 
    method getBalance(){ return balance }
    
    ghost balance(acc){
        if (this.myAccount == acc)  then  this.balance
             else  ( if this.next==null then -1 else this.next.find(acc) )
    }
}          
\end{lstlisting}
 \vspace*{-7mm}
\caption{$M_{BA2}$: Implementation of \prg{Node}   -- version 2}
\label{fig:BanAccImplV2n}
\end{figure}

\begin{figure}[htb]
\begin{lstlisting}
class Account{

    field myBank
    
    method deposit{src,amt){
             this.myBank.deposit(this,src,amt)
    }   }    
    method makeAccount(amt){
      if (amt>=0 && this.balance>=amt) then{
           newAcc = this.myBank.makeNewAccount(0)
           newAcc.deposit(this,amt)
           return newAcc
    }    }     

   ghost balance(){ this.myBank.balance(this) }   
}
\end{lstlisting}
 \vspace*{-7mm}
\caption{$M_{BA2}$: Implementation of  \prg{Account} -- version 2}
\label{fig:BanAccImplV2a}
\end{figure}

$\M_{BA1}$, the second version is shown Fig. \ref{fig:BanAccImplV2a} and \ref{fig:BanAccImplV2b}. It keeps all the information 
in the \prg{ledger}: each \prg{Node} points to an \prg{Account} 
and contains the balance for this particular \prg{Account}. Here \prg{balance} is a
ghost field of \prg{Account}; the body of that declaration calls the ghost field function \prg{balanceOf} of the \prg{Bank} which in its
turn calls the ghost field function \prg{balanceOf} of the \prg{Node}. Note that the latter is recursively defined.

Note also that \prg{Node} exposes the function \prg{addToBalance(...)}; a call to this function   modifies the \prg{balance} of an \prg{Account} without requiring that the caller has access to the \prg{Account}. This might look as if it contradicted assertions (1) and  (2)
  from section \ref{sect:motivate:Bank}. However, upon closer inspection, we see that the assertion is satisfied. Remember that we employ a two-module semantics, where any change in the balance of an account is observed from one external state, to another external state. By definition, a configuration is external if its receiver is external.  However, no external object will ever have access to a \prg{Node}, and therefore no external object will ever be able to call the method \prg{addToBalance(...)}. In fact, we can add another assertion, (4), which promises that any internal object which is externally accessible is either a \prg{Bank} or an \prg{Account}.

(4)\ \  $\triangleq$\ \ $\forall \prg{o},\forall \prg{o}'.[\ \ \External{ \prg{o}}\  \wedge\ \neg (\External{ \prg{o}'}) \ \wedge\    \CanAccess{\prg{o}}{\prg{o}'}$\\
  $\strut \hspace{5.6cm}    
    \longrightarrow \ \    
  [\    \prg{o}:\prg{Account}\ \vee \ \prg{o}':\prg{Bank}\  \ ]\ \ \  \ \ \ \ ] \hfill $

%% file: ERC20.tex
 
 ERC20~\cite{ERC20} is a widely used token standard which describes the 
 basic functionality expected by any    Ethereum-based token contract. 
 It issues and keeps track of participants' tokens, and supports the  transfer
 of tokens between participants. 
%
%
Transfer of tokens 
 can   take place only provided that  there were sufficient tokens in the
 owner's account, and that
 the transfer was instigated by the owner, or by somebody authorized
 by the owner.

We specify this in \Chainmail as follows:
A decrease in  a participant's \prg{balance} 
can only be caused by a transfer instigated by the 
account holder themselves\\ (\ie $\Calls {\prg{p}} {\prg{transfer}} {...} {...}$), or by
an authorized transfer instigated by another participant $\prg{p}''$  (\ie $\Calls {\prg{p}''} {{\prg{transferFrom}} } {..} {..}$) who 
has authority for more than the tokens spent (\ie  $\prg{e}.\prg{allowed}(\prg{p},\prg{p}'')\geq \prg{m}$)
 
\vspace{.15cm}
\noindent
$\forall \prg{e}:\prg{ERC20}.\forall \prg{p}:\prg{Object}.\forall \prg{m},\prg{m}':\prg{Nat}.$\\
\strut \hspace{0.3cm} $[\ \ \prg{e}.\prg{balance(p)}=\prg{m}+\prg{m'}\ \wedge \ \Next{\prg{e}.\prg{balance(p)}=\prg{m}'}$ \\ 
\strut \hspace{0.4cm} \ \ \ $\longrightarrow$\\
\strut \hspace{0.4cm} \ \ \ $\exists \prg{p}',\prg{p}'':\prg{Object}.$ \\
\strut \hspace{0.4cm} \ \ \  $[\ \  \Calls{\prg{p}} {\prg{transfer}}  {\prg{e}}  {\prg{p}',\prg{m}} \  \  \ \vee\, $\\
\strut \hspace{0.4cm} \ \ \   $\ \ \ \ \prg{e}.\prg{allowed}(\prg{p},\prg{p}'')\geq \prg{m} \ \wedge \ \Calls{\prg{p}''} {\prg{transferFrom}}  {\prg{e}}  {\prg{p}',\prg{m}}\       \  ]$\\
\strut \hspace{0.3cm} $] $
\vspace{.15cm}

\noindent
That is to say: if next configuration witnesses a decrease of \prg{p}'s balance by
 $\prg{m}$, then the current configuration was a call of \prg{transfer} instigated by
 \prg{p}, or  a call of \prg{transferFrom} instigated by somebody authorized by \prg{p}.
 The term $\prg{e}.\prg{allowed}(\prg{p},\prg{p}'')$,  means that the
ERC20 variable \prg{e} holds a field called \prg{allowed}   which maps pairs of participants to numbers; such
mappings are supported in Solidity\cite{Solidity}.
 
We now define what it means for $\prg{p}'$ to be authorized  to  spend 
up to \prg{m} tokens on  $\prg{p}$'s behalf: At some point in the
past,  \prg{p} gave authority to $\prg{p}'$  to spend   \prg{m}
plus the sum of  tokens
spent so far by $\prg{p}' $ on the behalf of \prg{p}.

\vspace{.15cm}
\noindent
 $\forall \prg{e}:\prg{ERC20}.\forall \prg{p},\prg{p'}:\prg{Object}.\forall \prg{m}:\prg{Nat}.$\\
\strut \hspace{0.3cm} $[\ \ \prg{e}.\prg{allowed}(\prg{p},\prg{p}')=\prg{m} $\\
\strut \hspace{0.4cm} \ \ \ $\longrightarrow$\\
\strut \hspace{0.4cm} \ \ \  
     $\PrevId\langle\ \  \Calls{\prg{p}}  {\prg{approve}}  {\prg{e}} {\prg{p}',\prg{m}} $\\
      \strut \hspace{1.7cm} \ $\vee $\\
\strut \hspace{1.7cm} \  
     $    \prg{e}.\prg{allowed}(\prg{p},\prg{p}')=\prg{m}   
        \  \wedge\ $\\
\strut \hspace{1.5cm} \ \ \ \ \          
$  \neg   (\, {\Calls{\prg{p}'} {\prg{transferFrom}} {\prg{e}} {\prg{p},\_}   } \, \vee \, {\Calls{\prg{p}} {\prg{approve}} {\prg{e}} {\prg{p},\_} } \, ) $\\
      \strut \hspace{1.7cm}\  $\vee $\\
\strut \hspace{1.7cm}   \  $ \exists \prg{p}'':\prg{Object}.\exists\prg{m'}:\prg{Nat}.$\\
 \strut \hspace{1.7cm}\  $[\   
  \prg{e}.\prg{allowed}(\prg{p},\prg{p}')=\prg{m}+\prg{m}'  \, \wedge\,   {\Calls{\prg{p}'} {\prg{transferFrom}} {\prg{e}} {\prg{p}'',\prg{m}'}  }   ]$\\
\strut \hspace{0.4cm} \ \ \  \ \ \  \ \ \ \ \ $\rangle $\\
\strut \hspace{0.3cm} $]$
\vspace{.15cm}
 
In more detail\  $\prg{p}'$ is allowed to spend 
up to \prg{m} tokens on their behalf of $\prg{p}$, if in the   previous step either a)
 \prg{p} made the call \prg{approve} on \prg{e} 
with arguments $\prg{p}'$ and \prg{m}, or b)  
$\prg{p}'$ was allowed to spend  up to \prg{m} tokens for $\prg{p}$
and did not transfer any of \prg{p}'s tokens, nor did \prg{p} issue a fresh authorization,
or c) \prg{p} was authorized for $\prg{m}+\prg{m}'$ and spent $\prg{m}'$. 
  
  \vspace{.1cm}
 
 Thus, the holistic specification gives to account holders an
 "authorization-guarantee": their balance cannot decrease unless they
 themselves, or somebody they had authorized, instigates a transfer of
 tokens. Moreover, authorization is {\em not} transitive: only the
 account holder can authorise some other party to transfer funds from
 their account: authorisation to spend from an account does not confer
 the ability to authorise yet more others to spend also.
 
 
 With traditional  specifications, to obtain the "authorization-guarantee", 
one would need to inspect the pre- and post- conditions of {\em all} the functions
in the contract, and determine which of the functions decrease balances, and which of the functions 
 affect authorizations.
 In the case of the \prg{ERC20}, one would have to inspect all eight such specifications
 (given in appendix \ref{ERC20:appendix}), 
 where only five are relevant to the question at hand.
 In the general case, \eg the DAO, the number of   functions which are unrelated
 to the question at hand can be very large.
  
More importantly, with traditional  specifications, nothing stops the next release of the contract to add, 
\eg, a method which allows participants to share their authority, and thus
violate the "authorization-guarantee", or even a super-user from skimming 0.1\% from each of the accounts.

%% file: ERC20.appendix.tex
 
We compare the holistic and the traditional specification of ERC20

As we said earlier,  the holistic specification gives to account holders an
 "authorization-guarantee": their balance cannot decrease unless they
 themselves, or somebody they had authorized, instigates a transfer of
 tokens. Moreover, authorization is {\em not} transitive: only the
 account holder can authorise some other party to transfer funds from
 their account: authorisation to spend from an account does not confer
 the ability to authorise yet more others to spend also.
 
 With traditional  specifications, to obtain the "authorization-guarantee", 
one would need to inspect the pre- and post- conditions of {\em all} the functions
in the contract, and determine which of the functions decrease balances, and which of the functions 
 affect authorizations.
In Figure \ref{fig:classicalERC20} we outline a traditional specification for the \prg{ERC20}.
We give two speficiations for \prg{transfer}, another two for \prg{tranferFrom}, and one for all 
the remaining functions. The  first specification says, \eg, that if  
 \prg{p} has sufficient tokens, and it calls \prg{transfer}, then the transfer will take place.  
The second specification says that  if \prg{p} has insufficient tokens, then 
the transfer will not take place (we assume that in this
specification language, any entities not mentioned in the pre- or post-condition 
are not affected).
 
 Similarly, we would have to give another two specifications to define the behaviour of 
if \prg{p''} is authorized and executes \prg{transferFrom}, then   the balance decreases. 
But they are {\em implicit} about the overall behaviour and the   {\em necessary} conditions,
e.g., what are all the possible actions that can cause a decrease of balance?

\begin{figure}   
\fbox{
$
\begin{array}{c}
 \prg{e}:\prg{ERC20}\ \wedge\  \prg{p},\prg{p''}:\prg{Object} 
  \wedge\ \prg{m},\prg{m}',\prg{m}'':\prg{Nat}\ \wedge\   \\
 \prg{e}.\prg{balance(p)} = \prg{m}+\prg{m}'\ \ \wedge\ \ \prg{e}.\prg{balance(p'')} = \prg{m}''\ \ \wedge\ \ \prg{this}=\prg{p} \\
   \{ \ \ \prg{e.transfer(p'',m')} \ \ \}   \\
    \prg{e}.\prg{balance(p)} = \prg{m}\ \ \wedge\ \ \prg{e}.\prg{balance(p'')} = \prg{m}''+\prg{m}'
\ \ \\
\ \ \\
  \prg{e}:\prg{ERC20}\ \wedge\  \prg{p},\prg{p'}:\prg{Object}  \wedge\ \prg{m},\prg{m}',\prg{m}'':\prg{Nat}\ \wedge\     \prg{e}.\prg{balance(p)} = \prg{m} \ \ \wedge \prg{m} <  \prg{m}'  \\
   \{ \ \ \prg{e.transfer(p',m')} \ \ \}   \\
  \prg{e}.\prg{balance(p)} = \prg{m}  
  \\
  \\
\prg{e}:\prg{ERC20}\ \wedge\  \prg{p},\prg{p'},\prg{p}'':\prg{Object} 
  \wedge\ \prg{m},\prg{m}',\prg{m}'',\prg{m}''':\prg{Nat}\ \wedge\   \\
 \prg{e}.\prg{balance(p)} = \prg{m}+\prg{m}'\ \ \wedge\ \ \prg{e}.\prg{allowed(p,p')}=\prg{m}'''+\prg{m}' \ \wedge\\
  \prg{e}.\prg{balance(p'')} = \prg{m}''\ \ \wedge\ \ \prg{this}=\prg{p'} \\
   \{ \ \ \prg{e.transferFrom(p',p'',m')} \ \ \}   \\
    \prg{e}.\prg{balance(p)} = \prg{m}\ \ \wedge\ \ \prg{e}.\prg{balance(p'')} = \prg{m}''+\prg{m}'
     \ \wedge\ \ \prg{e}.\prg{allowed(p,p')}=\prg{m}'''
\ \ \\
\ \ \\
  \prg{e}:\prg{ERC20}\ \wedge\  \prg{p},\prg{p'}:\prg{Object}  \wedge\ \prg{m},\prg{m}',\prg{m}'':\prg{Nat}\ \wedge\ \prg{this}=\prg{p}' \ \wedge \\
      ( \ \prg{e}.\prg{balance(p)} =\prg{m} \wedge \prg{m} <  \prg{m}''\  \vee \ 
  \prg{e}.\prg{allowed(p,p')}=\prg{m'} \wedge \prg{m'} < \prg{m}'' \ ) \\
   \{ \ \ \prg{e.transferFrom(p,p'',m'')} \ \ \}   \\
  \prg{e}.\prg{balance(p)} = \prg{m} \wedge  \prg{e}.\prg{allowed(p,p')}=\prg{m'}
  \\
  \\
  \prg{e}:\prg{ERC20}\ \wedge\  \prg{p},\prg{p'}:\prg{Object}  \wedge\ \prg{m}:\prg{Nat}\ \wedge\ \prg{this}=\prg{p}  \\
   \{ \ \ \prg{e.approve(p',m')} \ \ \}   \\
  \prg{e}.\prg{allowed(p,p')} = \prg{m} 
  \\
  \\
   \prg{e}:\prg{ERC20}\ \wedge\ \prg{m}:\prg{Nat}\ \wedge\    \prg{p}.\prg{balance}=\prg{m}    \\
   \{ \ \ \prg{k}=\prg{e.balanceOf(p)} \ \ \}   \\
  \prg{k}=\prg{m} \ \wedge \ \prg{e.balanceOf(p)} = \prg{m}  
  \\
  \\
   \prg{e}:\prg{ERC20}\ \wedge\ \prg{m}:\prg{Nat}\ \wedge\    \prg{e}.\prg{allowed(p,p')}=\prg{m}    \\
   \{ \ \ \prg{k}=\prg{e.allowance(p,p')} \ \ \}   \\
  \prg{k}=\prg{m} \ \wedge \ \prg{e}.\prg{allowed(p,p')}=\prg{m} 
  \\
  \\
   \prg{e}:\prg{ERC20}\ \wedge\ \prg{m}:\prg{Nat}\ \wedge\     \sum_{\prg{p}\in dom(\prg{e}.\prg{balance})}^{}{\prg{e}.\prg{balance}(\prg{p})}=\prg{m}    \\
   \{ \ \ \prg{k}=\prg{e.totalSupply()} \ \ \}   \\
  \prg{k}=\prg{m}   
\end{array}
$
}
\caption{Classical specification for the \prg{ERC20}}
\label{fig:classicalERC20}
\end{figure}

%% file: DAO.tex
The DAO  {(Decentralised Autonomous Organisation)}~\cite{Dao}  is a famous Ethereum contract  which aims to support
collective management of funds,  and to place power directly in the
hands of the owners of the DAO
rather than delegate it to directors.
Unfortunately, the DAO was not robust:
a re-entrancy bug   exploited in June 2016 led  to a loss of   \$50M, and
a hard-fork in the  chain ~\cite{DaoBug}.
%
With holistic specifications  we can  write a succinct requirement that a
DAO contract should always be able to repay any owner's money.
Any contract which satisfies such a holistic specification cannot demonstrate the DAO bug.
 
Our specification consists of three requirements.
First, that the DAO always holds at least as 
much money as any owner's balance.
To express this we use 
the field \prg{balances} which is a mapping from participants's addresses to 
numbers. Such mapping-valued fields exist in Solidity, but they could
also be taken to be ghost fields~\cite{ghost}.
  
\vspace{.1cm}

\noindent
 \strut \hspace{0.5cm} $\forall \prg{d}:\prg{DAO}.\forall \prg{p}:\prg{Any}.\forall\prg{m}:\prg{Nat}.$\\
\strut \hspace{0.5cm} $[\ \ \prg{d.balances(p)}=\prg{m}  \ \ \  \longrightarrow  \ \ \ \prg{d}.\prg{ether}\geq \prg{m} \ \ ] $

\noindent
Second, that when an owner asks to be repaid, she is sent all her money.
\vspace{.1cm}

\noindent
 \strut \hspace{0.5cm} $\forall \prg{d}:\prg{DAO}.\forall \prg{p}:\prg{Any}.\forall\prg{m}:\prg{Nat}.$\\
\strut \hspace{0.5cm} $[\ \ \prg{d.balance(p)}=\prg{m}
 \ \wedge \ \Calls{\prg{p}}{\prg{repay}}{\prg{d}}{\_}  $\\
 $\strut \hspace{5.5cm}   \ \ \  \longrightarrow  \ \ \  \Future{\Calls{\prg{d}}{\prg{p}}{\prg{send}}{\prg{m}}}\ \ ] $  
\vspace{.1cm}

%

\noindent
\sd{Third, that the balance of an owner is a function of the its balance in the previous step,
or the result of it joining the DAO, or asking to be repaid \etc.}
 
\noindent
$\strut \hspace{0.5cm} \forall \prg{d}:\prg{DAO}.\forall \prg{p}.\forall:\prg{m}:\prg{Nat}.$\\
$\strut \hspace{0.5cm} [ \ \ \  \prg{d.Balance(p)}=\prg{m} \ \ \  \longrightarrow   
 \ \  \ \ 
  [ \  \ \Prev{\Calls{\prg{p}}{\prg{repay}}{\prg{d}}{\_}}\, \wedge\, \prg{m}=\prg{0} \ \ \ \ \vee $\\
$\strut \hspace{5.7cm}      
\Prev{\Calls{\prg{p}}{\prg{join}}{\prg{d}}{\prg{m}}}  \ \ \ \ \vee   $\\
 $\strut \hspace{5.7cm}  ... \  ]$ \\
%
$\strut \hspace{0.5cm} ] $

More cases are needed to reflect the financing and repayments of proposals, but they can be expressed with the concepts described so far.

\noindent
The requirement that \prg{d} holds at least \prg{m} ether precludes the DAO bug,
in the sense that  any contract satisfying that spec cannot exhibit  the  bug:   a contract
which satisfies the spec  is guaranteed to always have enough money to satisfy all \prg{repay} requests.
This guarantee  holds, regardless of how many functions there are in the DAO.
In contrast, to preclude the DAO  bug with a classical spec, one would need to write a spec for each of the
DAO functions (currently 19), a spec for each function of the auxiliary contracts used by the DAO,
and then study their emergent  behaviour.

These 19 DAO functions   have several different concerns:
who may vote   for a proposal, who is eligible to submit a proposal,
how long the consultation period is for deliberating a proposal, what
is the quorum, how to chose curators, what is the value of a token,
Of these groups of functions, only  a handful affect the balance of a
participant. Holistic specifications allow us to concentrate on aspect of DAO's behaviour across \emph{all} its functions.

%% file: DOM.tex
\emph{Attenuation} is the ability to provide to third party objects \emph{restricted}  access to an object's functionality. This is usually achieved through the introduction of an intermediate object. While such intermediate objects are a common programming practice, 
the term was coined, and the practice  was studied in detail in the object capabilities literature,
e.g.  \cite{MillerPhD}. 

The key structure underlying a web browser is the Domain Object Model
(DOM), a recursive composite tree structure of objects that represent
everything display in a browser window.  Each window has a single DOM
tree which includes both the page's main content and also third party
content such as advertisements. To ensure third party content cannot
affect a page's main content,
specifications for attenuation for the DOM were proposed in
\textit{Devriese et al:}   \cite{dd}. 

This example deals with a tree of DOM nodes: Access to a DOM node
gives access to all its parent and children nodes, and the ability to
modify the node's properties. However, as the top nodes of the tree
usually contain privileged information, while the lower nodes contain
less crucial third-party information, we want to be able to limit  access given to third parties to only the lower part of the DOM tree. We do this through a \prg{Wrapper}, which has a field \prg{node} pointing to a \prg{Node}, and a field \prg{height} which restricts the range of \prg{Node}s which may be modified through the use of the particular \prg{Wrapper}. Namely, when you hold a \prg{Wrapper}  you can modify the \prg{property} of all the descendants of the    \prg{height}-th ancestors of the \prg{node} of that particular \prg{Wtrapper}. 

In Figure \ref{fig:WrapperUse} we show an example of the use of  \prg{Wrapper} objects attenuating the use of \prg{Node}s  The function \prg{usingWrappers} takes as parameter an object of unknown provenance, here called \prg{unknwn}. On lines 2-7 we create  a tree consisting of nodes \prg{n1}, \prg{n2}, ... \prg{n6}, depicted as blue circles on the   right-hand-side of the Figure. On line 8 we create a wrapper of \prg{n5} with height \prg{1}. This means that the wrapper \prg{w} may be used to modify \prg{n3}, \prg{n5} and \prg{n6} (\ie the objects in the green triangle), while it cannot be used to modify \prg{n1}, \prg{n2}, and \prg{4} (\ie the objects within the blue triangle). 
On line 8 we call a  function named \prg{untrusted} on the \prg{unknown} object, and pass \prg{w} as   argument. 

\begin{figure}[htb]
\begin{tabular}{llll}
\ \ &
\begin{minipage}{0.45\textwidth}
\begin{lstlisting}
method usingWrappers(unknwn){
   n1=Node(null,"fixed"); 
   n2=Node(n1,"robust"); 
   n3=Node(n2,"const"); 
   n4=Node(n3,"volatile");
   n5=Node(n4,"variable");
   n6=Node(n5,"ethereal");
   w=Wrapper(n5,1);
   
   unknwn.untrusted(w);
   
   assert n2.property=="robust" 
   ...
}
\end{lstlisting}
\end{minipage}
& & 
\begin{minipage}{0.75\textwidth}
\includegraphics[width=\linewidth, trim=145  320 60 105,clip]{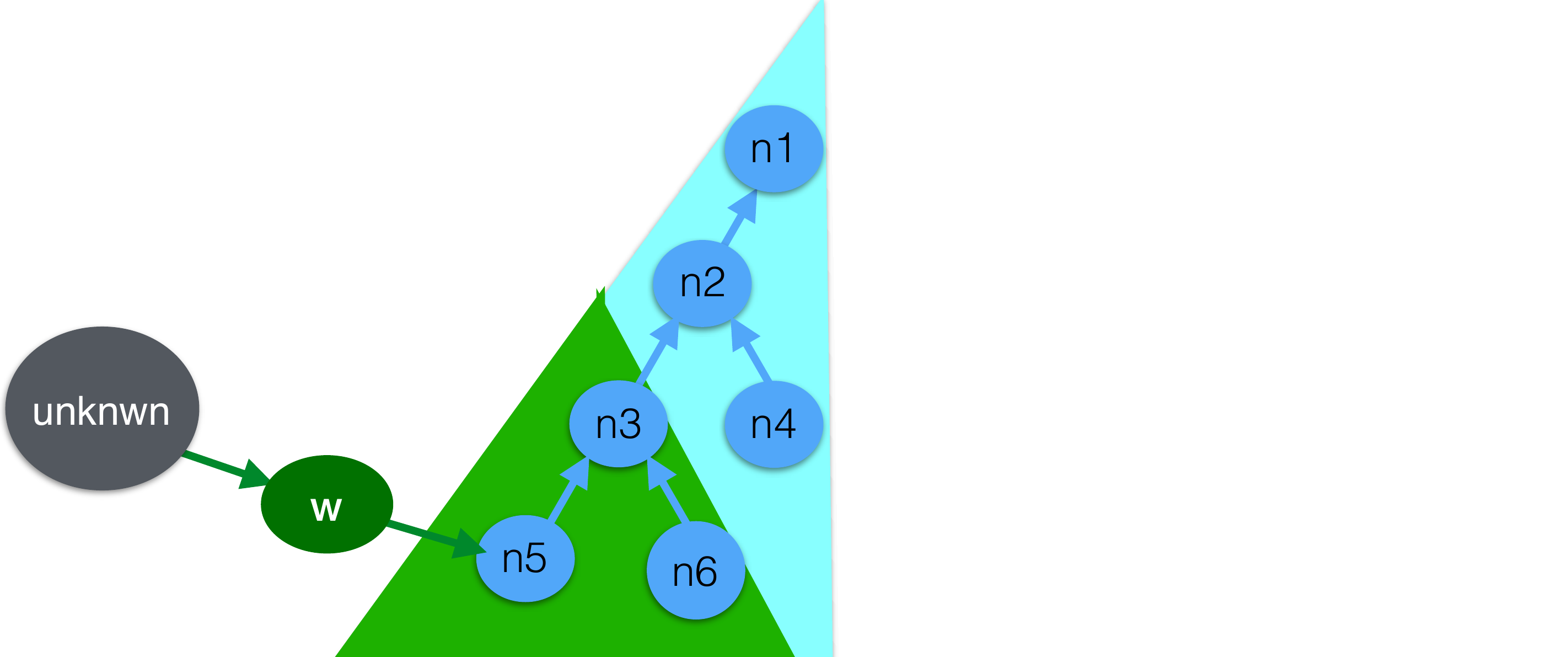}
\strut \\
\strut \\

\end{minipage}
\end{tabular}
 \vspace*{-4.5mm}
\caption{\prg{Wrapper}s protecting \prg{Node}s }
\label{fig:WrapperUse}
\end{figure}

Even though we know nothing about the \prg{unknown} object or its \prg{untrusted} function, and even though the call gives to \prg{unknown}  access to \prg{w}, which in turn has transitive access to all  \prg{Node}-s in the tree, 
we know that 
line 100 will not affect the   \prg{property} fields of the nodes \prg{n1}, \prg{n2}, and \prg{n4}. 
Thus, the assertion on line 12 is guaranteed to succeed. 
The question is how do we specify \prg{Wrapper}, so as to be able to make such an argument. 

A  specification of the class \prg{Wrapper} 
in the traditional   style, \eg  \cite{Leavens-etal07} 
consists of pairs of pre- and post- conditions for each of the
functions of that class. Each such pair gives a {\em sufficient}
condition for some effect to take place: for example the
call \prg{w.setProperty(i,prp)} where \prg{i} is smaller
than \prg{w.height} is a sufficient condition to
modify   \prg{property} of the \prg{i}-th parent of \prg{w.node}. But
we do not know what other ways there may be  to modify a
node's  \prg{property}. In other words, we have not specified
the \emph{necessary conditions}.
%
%
%
In our example:

\begin{quote}
The \emph{necessary} condition for the modification of \prg{nd.property} for some \prg{nd} of class \prg{Node}  is either access to some   \prg{Node} in the same tree, or  access to a \prg{w} of class \prg{Wrapper} where the \prg{w.height}-th parent of \prg{w} is an ancestor of \prg{nd}.
\end{quote}

With such a specification we can prove that the assertion on line 12 will succeed. And, more importantly, we can ensure that all future updates of the \prg{Wrapper} abstract data type will uphold the \emph{protection} of the \prg{Node} data.
%
%
To give a flavour of \Chainmail, we use it  express the requirement from above:
\vspace{.1cm}

\noindent
$\forall \prg{S}:\prg{Set}.\forall \prg{nd}:\prg{Node}.\forall \prg{o}:\prg{Object}.$\\
$[$ \\
$\strut\  \ {\Using{\Future {\Changes {\prg{nd.property}}}}  {\prg{S}}}$ \\
$\strut  \ \ \longrightarrow$\\
$\strut \ \ \ \exists \prg{o}.[\ \prg{o}\in\prg{S}\ \ \wedge\ \ \neg(\prg{o}:\prg{Node})\ \ \wedge\  \ \neg(\prg{o}:\prg{Wrapper})\ \ \ \wedge \  $\\
$ \strut\ \ \  \ \ \ \ \ \ \ \ [\ \exists \prg{nd}':\prg{Node}. \CanAccess{\prg{o}}{\prg{nd}'}\  \ \ \  \vee$\\
$ \strut\ \  \ \  \ \  \ \ \ \ \ \ \ \exists \prg{w}:\prg{Wrapper}.\exists k\!:\!\mathbb{N}.
  (\ \CanAccess{\prg{o}}{\prg{w}}  \ \wedge\ \prg{nd}.\prg{parnt}^k\!=\!\prg{w.node}.\prg{parnt}^{\prg{w.height}}) \ \ \ ]\ ]$\\
$ ]$

\vspace{.1cm}


\noindent
That is, if  the value of \prg{nd.property} is modified ($\Changes{\_}$) at some future point ($\Future{\_}$) 
and if reaching that future point involves no more objects than those from
set \prg{S} (\ie ${\Using{\_} {\prg{S}}}$), then at least one (\prg{o}) of the objects in \prg{S} is not a \prg{Node} nor a \prg{Wrapper}, and \prg{o} has direct access to some  node ($\CanAccess{\prg{o}}{\prg{nd}'}$), or  to some wrapper \prg{w} and the \prg{w.height}-th parent of \prg{w} is an ancestor of \prg{nd} (that is, $\prg{parnt}^k\!=\!\prg{w.node}.\prg{parnt}^{\prg{w.height}}$).
 Note that our ``access'' is intransitive: $\CanAccess x y$ holds if  either \prg{x} has a field  pointing to \prg{y}, or  \prg{x}  is the receiver and \prg{y} is one of the arguments  in the executing method call.
 

%% file: COQ.tex
 In this section we present the properties of Chainmail that have been formalised in the Coq model. Table \ref{Coq} refers to proofs that can be found in the associated Coq formalism \cite{coq}.

\begin{table}
  \begin{tabular}{|l|l|l|l|}
    \hline

\textbf{Lemma \ref{lemma:linking} (and 3)} &
Properties of Linking
        & 
\parbox{.45\textwidth}{\scriptsize\begin{enumerate}[label={(\arabic*)}]
            \item \texttt{moduleLinking\_associative}
            \item \texttt{moduleLinking\_commutative\_1}
            \item \texttt{moduleLinking\_commutative\_2}
            \item \texttt{linking\_preserves\_reduction}
        \end{enumerate}}
        \\
\hline
\textbf{Lemma \ref{lemma:classic}} &   
\parbox{.45\textwidth}{\scriptsize\begin{enumerate}[label={(\arabic*)}]
            \item $A \wedge \neg A \equiv \texttt{false}$
            \item $A \vee \neg A \equiv \texttt{true}$
            \item $A \vee A' \equiv A' \wedge A$
            \item $A \wedge A' \equiv A' \wedge A$
            \item $(A \vee A') \vee A'' \equiv A \vee (A' \vee A'')$
        \end{enumerate}}
        & 
\parbox{.45\textwidth}{\scriptsize\begin{enumerate}[label={(\arabic*)}]
            \item \texttt{sat\_and\_nsat\_equiv\_false}
            \item -
            \item \texttt{and\_commutative}
            \item \texttt{or\_commutative}
            \item \texttt{or\_associative}
        \end{enumerate}}
        \\
\hline
\textbf{Lemma \ref{lemma:basic_assertions_classical}} &   
\parbox{.45\textwidth}{\scriptsize\begin{enumerate}[label={(\arabic*)}]
            \item $A \wedge \neg A \equiv \texttt{false}$
            \item $A \vee \neg A \equiv \texttt{true}$
            \item $A \vee A' \equiv A' \wedge A$
            \item $A \wedge A' \equiv A' \wedge A$
            \item $(A \vee A') \vee A'' \equiv A \vee (A' \vee A'')$
            \item $(A \vee A') \wedge A'' \equiv (A \vee A'') \wedge (A' \vee A'')$
            \item $(A \wedge A') \vee A'' \equiv (A \wedge A'') \vee (A' \wedge A'')$
            \item $\neg (A \wedge A') \equiv (\neg A \vee \neg A')$
            \item $\neg (A \vee A') \equiv (\neg A \wedge \neg A')$
            \item $\neg (\exists x.A) \equiv \forall x. (\neg A)$
            \item $\neg (\exists S.A) \equiv \forall S. (\neg A)$
            \item $\neg (\forall x.A) \equiv \exists x. (\neg A)$
            \item $\neg (\forall S.A) \equiv \exists S. (\neg A)$
        \end{enumerate}}
        & 
\parbox{.45\textwidth}{\scriptsize\begin{enumerate}[label={(\arabic*)}]
            \item \texttt{sat\_and\_nsat\_equiv\_false}
            \item -
            \item \texttt{and\_commutative}
            \item \texttt{or\_commutative}
            \item \texttt{or\_associative}
            \item \texttt{and\_distributive}
            \item \texttt{or\_distributive}
            \item \texttt{neg\_distributive\_and}
            \item \texttt{neg\_distributive\_or}
            \item \texttt{not\_ex\_x\_all\_not}
            \item \texttt{not\_ex\_$\Sigma$\_all\_not}
            \item \texttt{not\_all\_x\_ex\_not}
            \item \texttt{not\_all\_$\Sigma$\_ex\_not}
        \end{enumerate}}
\\
\hline
  \end{tabular}
  \caption{Chainmail Properties Formalised in Coq}
  \label{Coq}
\end{table}